\documentclass[10pt,letterpaper]{amsart}
\pdfoutput=1	
\usepackage[utf8]{inputenc}

\usepackage{hyperref}

\usepackage{amsmath}	
\usepackage{amssymb}	
\usepackage{amsfonts}
\usepackage{amsthm}				

\usepackage{appendix}				

\usepackage{xcolor}	

\usepackage{soul}				
\sethlcolor{yellow}				

\usepackage{graphicx}			
\usepackage{subcaption}	
\graphicspath{{Figures/}}

\usepackage{media9}

\usepackage{float}

\usepackage[backend=bibtex,style=numeric,citestyle=numeric,maxbibnames=99,giveninits,doi=false,isbn=false,url=false,eprint=false]{biblatex}
\DeclareFieldFormat[article,periodical]{number}{\bibstring{number}~#1}

\renewbibmacro{in:}{\ifentrytype{article}{}{\printtext{\bibstring{in}\intitlepunct}}}
\renewbibmacro*{journal+issuetitle}{%
  \usebibmacro{journal}%
  \setunit*{\addspace}%
  \iffieldundef{series}
    {}
    {\newunit
     \printfield{series}%
     \setunit{\addspace}}%
  \printfield{volume}%
  \setunit{\addspace}%
  \usebibmacro{issue+date}%
  \setunit{\addcomma\space}%
  \printfield{number}%
  \setunit{\addcolon\space}%
  \usebibmacro{issue}%
  \setunit{\addcomma\space}%
  \printfield{eid}
  \newunit}

\allowdisplaybreaks

\renewcommand{\vec}[1]{\boldsymbol{#1}}
\let\oldhat\hat
\renewcommand{\hat}[1]{\oldhat{\boldsymbol{#1}}}

\newcommand{\dee}{\ensuremath{\textrm{d}}}

\newcommand{\fdf}[1]{\ensuremath{ \frac{\dee}{\dee #1}}}

\newcommand{\inty}[4]{\ensuremath{ \int_{#1}^{#2} \! #3 \, \dee#4 }}

\newcommand{\field}[1]{\mathbb{#1}}
\newcommand{\ip}[2]{\ensuremath{ \left< \left. #1 \right| #2 \right> } }
\newcommand{\bra}[1]{\ensuremath{ \left< #1 \right| } }
\newcommand{\ket}[1]{\ensuremath{ \left| #1 \right> } }

\newtheorem{definition}{Definition}		
\newtheorem{assumption}{Assumption}
\newtheorem{remark}{Remark}
\newtheorem{theorem}{Theorem}
\newtheorem{lemma}{Lemma}
\newtheorem{proposition}[lemma]{Proposition}

\numberwithin{lemma}{section}
\numberwithin{example}{section}
\numberwithin{equation}{section}
\numberwithin{remark}{section}
\numberwithin{definition}{section}
\numberwithin{assumption}{section}
\numberwithin{figure}{section}
\newtheorem*{definition*}{Definition}		
\newtheorem*{assumption*}{Assumption}
\newtheorem*{remark*}{Remark}
\newtheorem*{theorem*}{Theorem}
\newtheorem*{lemma*}{Lemma}
\newtheorem*{proposition*}{Proposition}
\newtheorem*{corollary*}{Corollary}
\newtheorem*{example*}{Example}

\let\Tr\relax
\DeclareMathOperator{\Tr}{Tr}

\title{Locality of the windowed local density of states }

\author{Terry A. Loring}
\address{Department of Mathematics and Statistics, University of New Mexico, Albuquerque, NM 87131, U.S.A.}
\email{loring@math.unm.edu}
\author{Jianfeng Lu}
\address{Departments of Mathematics, Chemistry, and Physics, Duke University, Durham, NC 27708, U.S.A.}
\email{jianfeng@math.duke.edu}
\author{Alexander B. Watson}
\address{Department of Mathematics, University of Minnesota, Minneapolis, MN 55455, U.S.A.}
\email{watso860@umn.edu}

\thanks{The work of T.~A.~L.~is based on work supported by the National Science Foundation under Grant No. DMS 1700102. The work of J.~L.~is supported in part by the U.S.~National Science Foundation via grant DMS-2012286 and the U.S.~Department of Energy via grant DE-SC0019449. The work of A.~B.~W.~is supported in part by ARO MURI Award W911NF-14-0247. We are grateful to Lucien Jezequel for pointing out an alternative proof which allowed for a considerably weakened regularity hypothesis \eqref{eq:f_regularity} in Lemma \ref{lem:locality_opnorm}.}

\date{\today}

\begin{document}

\begin{abstract}
We introduce a generalization of local density of states which is ``windowed'' with respect to position and energy, called the windowed local density of states (wLDOS). This definition generalizes the usual LDOS in the sense that the usual LDOS is recovered in the limit where the position window captures individual sites and the energy window is a delta distribution. We prove that the wLDOS is local in the sense that it can be computed up to arbitrarily small error using spatial truncations of the system Hamiltonian. Using this result we prove that the wLDOS is well-defined and computable for infinite systems satisfying some natural assumptions. We finally present numerical computations of the wLDOS at the edge and in the bulk of a ``Fibonacci SSH model'', a one-dimensional non-periodic model with topological edge states.
\end{abstract}

\maketitle

\section{Introduction}

The density of states (DOS) is a fundamental concept in condensed matter physics which is crucial for understanding electronic conductivity properties of materials. Roughly speaking, the DOS is the density of electronic states available to be occupied by an electron as a function of energy (ignoring electron-electron interactions). Mathematically, the DOS is the density of eigenvalues of the single-particle electronic Hamiltonian viewed as a function of the spectral parameter. The \emph{local} density of states (LDOS) is the contribution to the DOS from each point in space so that the average of the LDOS over all space equals the DOS at that energy. The LDOS has been used to clarify many phenomena in condensed matter physics. It is an especially important tool for studying systems without translational symmetry such as crystalline materials near defects or edges, disordered materials, and quasicrystals, where the Hamiltonian cannot be diagonalized using Bloch theory.

In this work we propose a generalization of the LDOS which is ``windowed'' with respect to position and energy, called the windowed local density of states (wLDOS). We start by defining the wLDOS for finite tight-binding systems, and in this context we show that the wLDOS reduces to the usual LDOS whenever the position window captures individual sites and the energy window is a delta distribution. We then prove that the wLDOS is local in the sense that it can be computed using a spatial truncation of the Hamiltonian to a neighborhood around each point of interest. Using locality of the wLDOS we then show that the wLDOS is well-defined and computable for a broad class of infinite tight-binding systems. 

We finally present a numerical study of the wLDOS in the bulk and near the edge of a ``quasi-crystalline'' SSH model: a model of a one-dimensional material with no spatial periodicity which nonetheless supports a non-trivial bulk topological invariant and associated edge states. Note that there is no fundamental reason to restrict to a one-dimensional model, since the definition and locality property of the wLDOS are \emph{dimension-independent}. However, in this work we restrict our numerical experiments to one spatial dimension because analogous numerical experiments in higher dimensions will be computationally more intensive and go beyond the scope of this work. 

We are motivated to introduce the wLDOS as an alternative to the standard LDOS for several reasons. One is physical, in that experimental data extracted in spectroscopy is necessarily blurred with respect to energy and position because of the finite resolution of experimental probes. 
This can be seen clearly in 
Figures \ref{fig:nanowire_interface_scan} and \ref{fig:carbon_nanotube_scan}. 
We therefore expect numerical computations of the wLDOS to more closely resemble experimental data than numerical computations of the LDOS. This can be seen in Figure \ref{fig:wLDOS-examples}, where we present numerical computations of the wLDOS for a one-dimensional periodic SSH model and our quasicrystal variant. The energy and position windows used in these computations are shown in Figure \ref{fig:Windows-for-examples}. In our in-depth numerical study in Section \ref{sec:numerical_expts} we use narrower position windows, but even the narrower windows, when placed over sites which are not uniformly spaced, can partially cover more than one site: see Figure \ref{fig:Windows-for-later}.

Another reason to compute the wLDOS rather than the LDOS is numerical. When approximating the DOS one must smooth somewhat in energy to avoid implicitly attempting to compute all of the eigenvalues \cite{lin2016approximating} and the same issue arises when approximating the LDOS. Finally, having the option of a variable spatial window may lead to more flexibility in how one parallelizes LDOS computations, and is convenient for plots of LDOS for systems with irregular positioning of sites and for systems with continuous degree of freedom.

\begin{figure}
\includegraphics{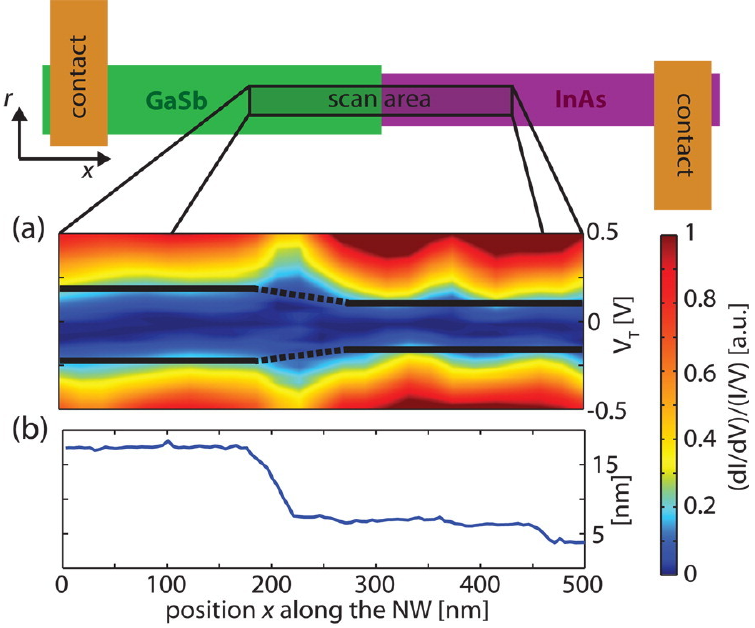}
\caption{An example of spectrocopy revealing the change in bandgap across an interface between two types of nanowire. Panel (a) is a contour plot of experimental measurement of the LDOS as a function of position along the nanowire (horizontal axis) and energy (vertical axis), while panel (b) shows the height of the sample as a function of position. Reprinted with permission from \cite{persson2015scanning} Persson, Olof, et al ``Scanning tunneling spectroscopy on InAs\textendash GaSb Esaki diode nanowire devices during operation'' \protect\href{https://pubs.acs.org/doi/full/10.1021/acs.nanolett.5b00898}{Nano letters 15.6 (2015): 3684-3691} Copyright 2015 American Chemical Society (further permission related
to the material excerpted should be directed to the ACS). \label{fig:nanowire_interface_scan}}
\end{figure}

\begin{figure}
\includegraphics{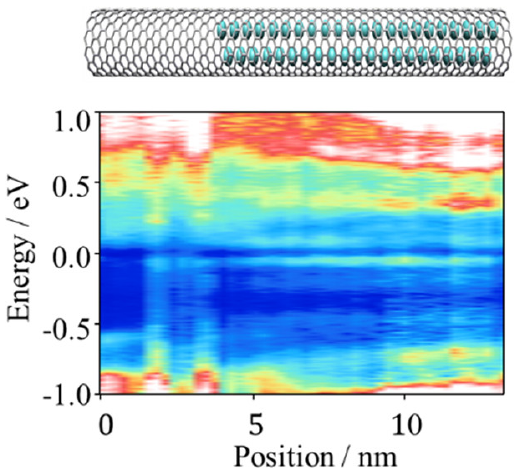}

\caption{An example of the output of scanning tunneling microscopy/spectroscopy measurement of the LDOS performed on an encapsulation of europium nanowires encapsulated in a carbon nanotube. Red and blue correspond to high and low densities respectively. The horizontal yellow streak within the blue area corresponds to localized states with energy within the band gap. Reprinted with permission from \cite{nakanishi2017modulation} Terunobu Nakanishi, Ryo Kitaura, Takazumi Kawai, Susumu Okada, Shoji Yoshida, Osamu Takeuchi, Hidemi Shigekawa, and Hisanori Shinohara, The Journal of Physical Chemistry C 2017, 121 (33), 18195\textendash 18201 DOI: 10.1021/acs.jpcc.7b04047. Copyright 2017 American Chemical Society. \label{fig:carbon_nanotube_scan}}
\end{figure}


\begin{figure}
\centering
\raisebox{4cm}{(a)}\includegraphics[scale=.75]{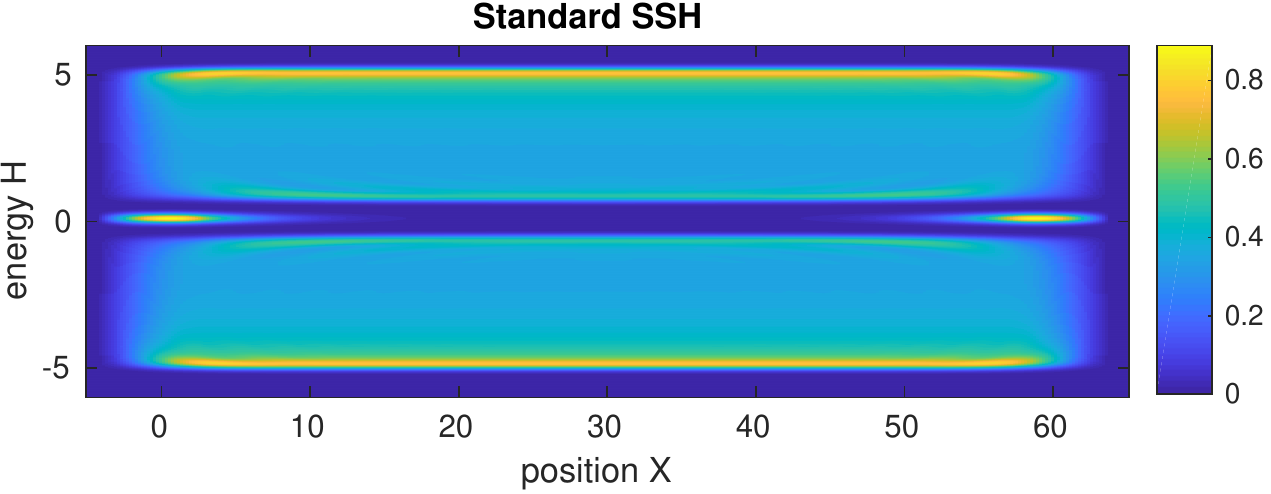} \\
\raisebox{4cm}{(b)}\includegraphics[scale=.75]{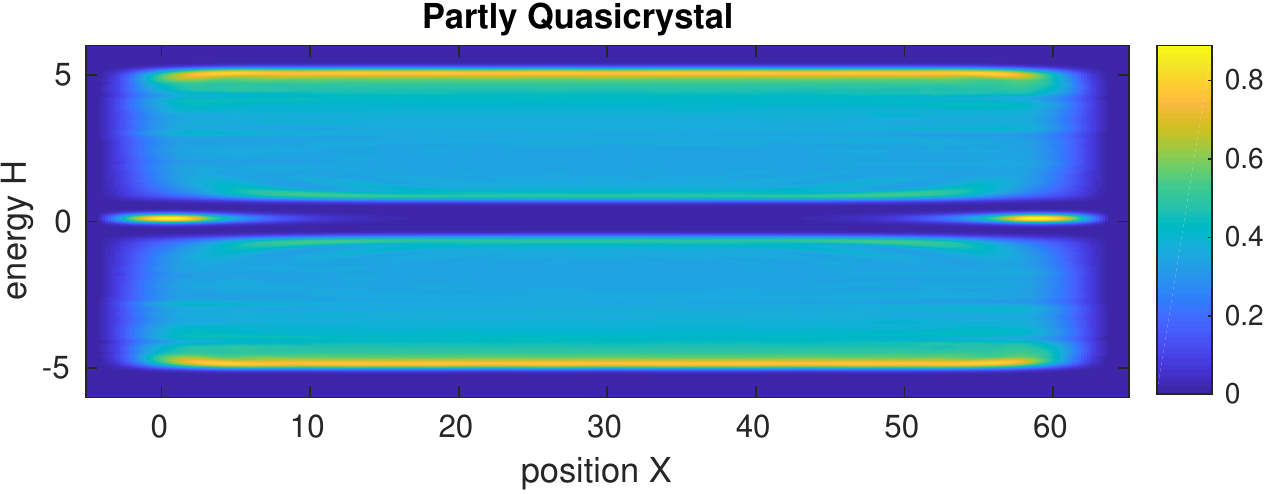} \\
\raisebox{4cm}{(c)}\includegraphics[scale=.75]{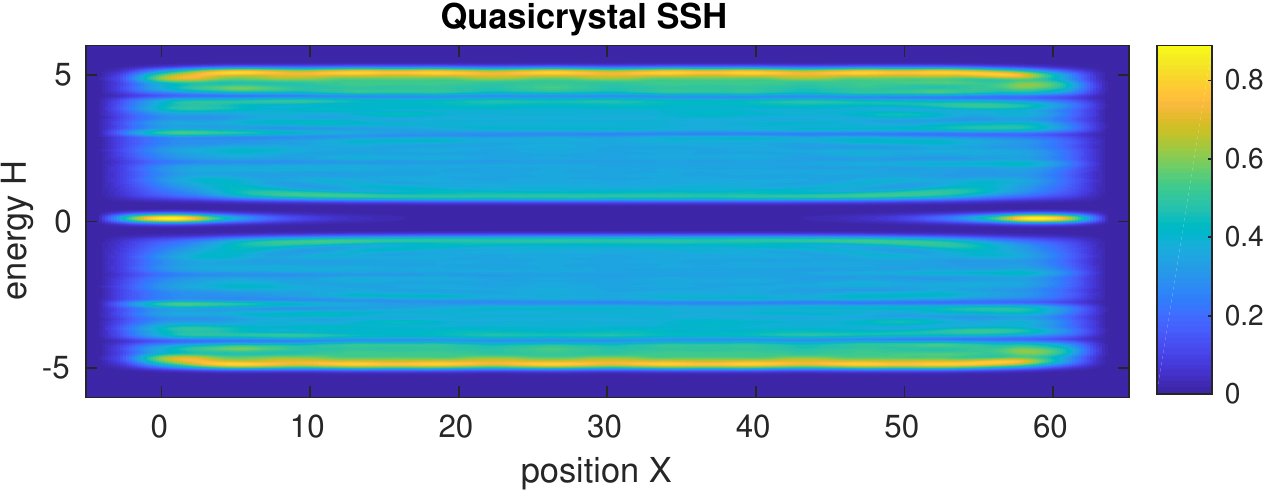} \\

\caption{Numerical computations of the wLDOS for (a) the periodic SSH model, (c) a quasicrystalline variant of the SSH model, and (b) an interpolation between those two models. The energy window is a narrow Gaussian, defined by \eqref{eq:eta_f} with $\eta^{-1} = 9$, while the position window is as in Section \ref{sec:numerical_expts} but scaled to be twice as wide, i.e. to have support $[-2,2]$. The computations show important local spectral features clearly: edge modes with energy in the bulk gap, and gap opening within the bulk bands due to the perturbation which breaks translation symmetry.
\label{fig:wLDOS-examples}}
\end{figure}

\begin{figure}
\centering
\raisebox{4cm}{(a)}\includegraphics[scale=.75]{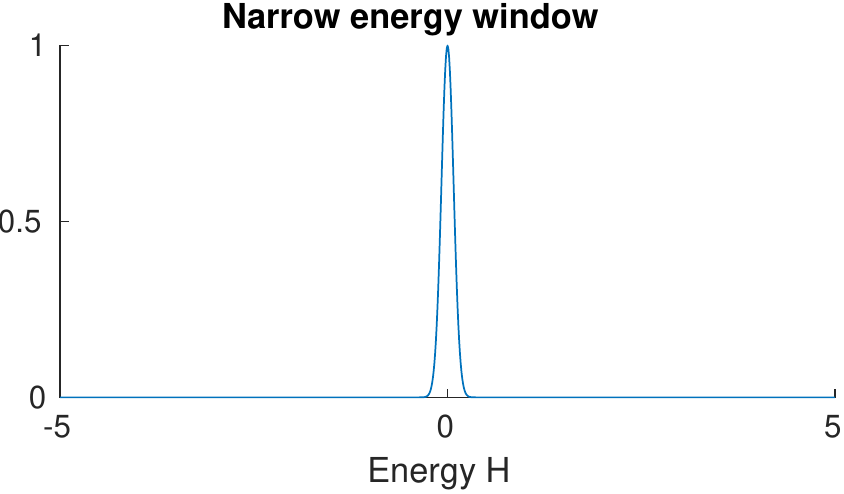} \\
\raisebox{4cm}{(b)}\includegraphics[scale=.75]{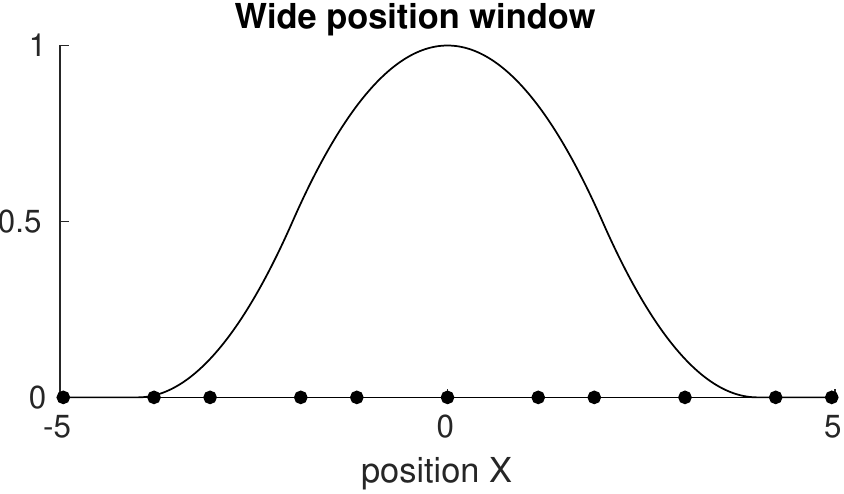} \\
\caption{The energy and position windows used for the wLDOS calculations shown in Figure~\ref{fig:wLDOS-examples}. Dots along the horizontal axis of (b) represent positions of sites in the quasicrystal SSH model. For the definition of the wLDOS and of the window functions, see Section \ref{sec:wLDOS_sec}. Taking a position window function which captures multiple sites models the finite resolution of experimental measurements of the LDOS.
\label{fig:Windows-for-examples}}
\end{figure}

\begin{figure}
\centering
\includegraphics[scale=.75]{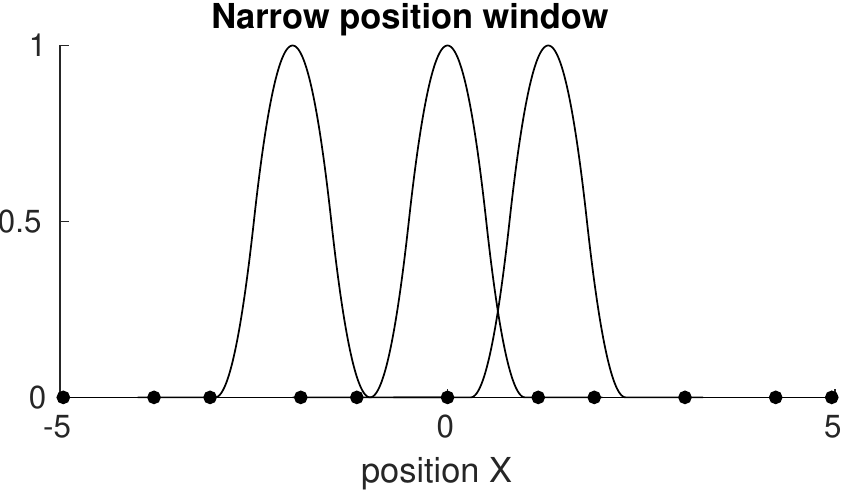}
\caption{The position windows used for the wLDOS calculations in Section \ref{sec:numerical_expts}. Dots along the horizontal axis represent positions of sites in the quasicrystal SSH model. Note that even the narrower position windows can overlap multiple sites. The wLDOS reduces to the LDOS only in the limit where each position window captures precisely one site.
\label{fig:Windows-for-later}}
\end{figure}

\subsection{Related literature}

The derivation of the LDOS is standard in textbooks on condensed matter physics, see for example \cite{kaxiras_joannopoulos_2019}. We do not attempt to summarize the physics literature on, or using, the LDOS, but recall some relevant mathematical works. Massatt, Luskin, and Ortner \cite{2017MassattLuskinOrtner} (see also \cite{doi:10.1137/17M1141035,Carr2017}) proposed a method for computing the DOS of an incommensurate bilayer system by averaging the LDOS over local atomic configurations. In the process they proved locality of the LDOS for that system using resolvent calculus. In contrast, we prove locality of the wLDOS (a generalization of the LDOS) for systems in arbitrary dimensions satisfying more general assumptions via a spectral flow argument.
The idea to use locality of quantum mechanical models to develop schemes for computing quantum mechanical observables is now well-established, see \cite{1999Goedecker,Weinan2010,WeinanLu2011,doi:10.1137/15M1022628,colbrook2019compute}. The electronic properties of one-dimensional aperiodic models have been considered in the physics literature, see \cite{Zhong1991,delaney1998local}.

\subsection{Outline of paper}

We recall usual definitions of the DOS and LDOS, and then define the wLDOS, in the relatively simple case of finite-dimensional tight-binding models in Section \ref{sec:wLDOS_def}. We will then prove locality of the wLDOS for such systems in Section \ref{sec:local_proof} before using this property to extend the definition to a class of infinite-dimensional tight-binding models in Section \ref{sec:infinite_systems}. We will then introduce the quasicrystalline SSH model in Section \ref{sec:Fib_SSH} and present results of our numerical experiments in Section \ref{sec:numerical_expts}.

\section{Windowed local density of states} \label{sec:wLDOS_def}

In this section we will recall standard definitions of the DOS and LDOS, and then introduce the windowed local density of states (wLDOS), in the simplest case of finite-dimensional tight-binding models.

\subsection{Finite-dimensional tight-binding models and standard definitions of the DOS and LDOS} \label{sec:standard_LDOS}

We start by establishing some notation for finite-dimensional tight-binding models. 

We consider sets of $N$ points (which we will refer to as sites) in $\field{R}^d$, with co-ordinates $\vec{x}_n = (x_1,...,x_d)$ for $1 \leq n \leq N$. We consider an electron with $M$ internal degrees of freedom hopping between these sites in the tight-binding approximation, so that the electronic wave-function is an element of the Hilbert space $\field{C}^N \otimes \field{C}^M \cong \field{C}^{\mathcal{N}}$ where $\mathcal{N} = M N$. For the moment we allow the Hamiltonian $H$ of the system to be an arbitrary $\mathcal{N} \times \mathcal{N}$ Hermitian matrix. Let $\delta_n^m \in \field{C}^{\mathcal{N}}$ denote the vector equalling $1$ at site $n$ and internal degree of freedom $m$, and zero in every other entry. For each $1 \leq l \leq d$, we define position operators by
\begin{equation} \label{eq:pos}
    X_l := \sum_{n = 1}^N x_n \sum_{m = 1}^M \ket{\delta_n^m} \bra{\delta_n^m} \quad 1 \leq l \leq d.
\end{equation}

\begin{remark} \label{rem:remark_on_continuum}
We expect that our results can be generalized to continuum PDE models, at the cost of some technical complications. One obvious difficulty is unboundedness of the Hamiltonian, since boundedness is essential at a few points in the present work, e.g. Lemma \ref{lem:locality_opnorm}, Assumption \ref{as:bounded}. We expect this difficulty can be overcome in the same way that we overcome the restriction to finite systems below, by using the fact that the wLDOS depends only on an energy truncation of the Hamiltonian (restriction of the Hamiltonian to the subspace of eigenstates with energy close to the energy of interest).
\end{remark}

The simplest definition of the DOS is directly as the distribution
\begin{equation} \label{eq:phi}
    D(E) := \frac{1}{\mathcal{N}} \sum_{j = 1}^{\mathcal{N}} \delta(\lambda_j - E)
\end{equation}
where $\lambda_j$ are the eigenvalues of $H$ counted with multiplicity \cite{kaxiras_joannopoulos_2019}. The LDOS is defined as follows. Let $\psi_j$, $j \in \{1,...,\mathcal{N}\}$, denote the eigenvector of $H$ corresponding to the eigenvalue $\lambda_j$ (if $\lambda_j$ is degenerate $\psi_j$ is not unique but any choice for $\psi_j$ will do) so that
\begin{equation}
    H \psi_j = \lambda_j \psi_j.
\end{equation}
We can then define using the spectral representation of $H$
\begin{equation} \label{eq:delta_H}
    \delta(H - E) = \sum_{j = 1}^{\mathcal{N}} \delta( \lambda_j - E ) \ket{\psi_j} \bra{\psi_j}.
\end{equation}
Taking the trace of \eqref{eq:delta_H} with respect to the basis $\{ \psi_j \}_{1 \leq j \leq \mathcal{N}}$ gives $\mathcal{N} D(E)$, and hence 
\begin{equation}
    D(E) = \frac{1}{\mathcal{N}} \Tr \delta( H - E ).
\end{equation}
Expanding the trace in the basis $\{ \delta_n^m \}_{1 \leq n \leq N, 1 \leq m \leq M}$  we derive
\begin{equation}
    D(E) = \frac{1}{\mathcal{N}} \sum_{n = 1}^N D_n(E),
\end{equation}
where
\begin{equation} \label{eq:LDOS}
    D_n(E) := \sum_{m = 1}^M \bra{ \delta_n^m } \delta(H - E) \ket{ \delta_n^m }
\end{equation}
is the LDOS defined at each site $n$ \cite{kaxiras_joannopoulos_2019}.


\begin{remark}
For the sake of clarity, we define the DOS and LDOS directly through the delta distribution. The DOS and LDOS can equivalently be defined via their action on test functions, see \cite{2017MassattLuskinOrtner} for example. 
\end{remark}

\subsection{The windowed DOS and LDOS} \label{sec:wLDOS_sec}

We will shortly introduce the main object of study in this work, the windowed LDOS (wLDOS), which is a version of the LDOS which is windowed in both energy and position. For simplicity, we first restrict attention to one spatial dimension. 

Let $f(\xi)$ be a positive function in $L^1(\field{R})$. For example, we can take $f$ to be a normalized Gaussian
\begin{equation} \label{eq:Gaussian}
    f(\xi) = \frac{1}{\sqrt{2 \pi} \sigma} e^{- \frac{ \xi^2 }{ 2 \sigma^2 } }
\end{equation}
for some $\sigma > 0$. Let $H$ be as in section \ref{sec:standard_LDOS}. We now define the windowed DOS (wDOS). 
\begin{definition} \label{def:wDOS}
We define the windowed DOS (wDOS) by 
\begin{equation} \label{eq:wDOS}
    W(E) := \frac{1}{\mathcal{N}} \sum_{j = 1}^N f(\lambda_j - E).
\end{equation}
\end{definition}
Following the argument of the previous section we derive that
\begin{equation} \label{eq:wDOS}
    W(E) = \frac{1}{\mathcal{N}} \sum_{m = 1}^M \sum_{n = 1}^N \bra{\delta_n^m} f(H - E) \ket{\delta_n^m}.
\end{equation}
We therefore define the windowed (in energy) LDOS by
\begin{definition} \label{def:wieLDOS}
We define the windowed (in energy) LDOS by
\begin{equation} \label{eq:wieLDOS}
    \sum_{m = 1}^M \bra{\delta_n^m} f(H - E) \ket{\delta_n^m}.
\end{equation}
\end{definition}
The wDOS and \eqref{eq:wieLDOS} clearly reduce to the standard DOS \eqref{eq:phi} and LDOS \eqref{eq:LDOS} in the limit where $f(\xi) \rightarrow \delta(\xi)$. 

We wish to consider a more general construction where the LDOS is windowed \emph{in position as well as energy}. To this end, let $g(\xi) \in L^1(\field{R})$ satisfy $0 \leq g \leq 1$ and be compactly supported.
For example, we can take $g$ to be a ``bump'' function such as 
\begin{equation} \label{eq:f} g(\xi) = \begin{cases} 0, & \xi \leq -2;
    \\ \frac{1}{2} (\xi + 2)^2, & -2 < \xi \leq -1; \\ 1 - \frac{1}{2}
    \xi^2, & -1 < \xi \leq 1; \\ \frac{1}{2} (\xi - 2)^2, & 1 < \xi \leq
    2; \\ 0, & 2 \leq \xi. \end{cases}
\end{equation}
Let $X$ be the one-dimensional position operator (recall \eqref{eq:pos}) 
\begin{equation} \label{eq:pos_R1}
    X = \sum_{n = 1}^N x_n \sum_{m = 1}^M \ket{ \delta_n^m } \bra{ \delta_n^m },
\end{equation}
where $x_n, 1 \leq n \leq N$ are the positions of each site. Then define for any real $E, x$
\begin{equation} \label{eq:mLDOS_1d}
    F_{E,x}(H,X) :=  g^{\frac{1}{2}}(X-x) f(H-E) g^{\frac{1}{2}}(X-x).
\end{equation}
We now define the one dimensional windowed LDOS (wLDOS) at energy $E$ and position $x$ with window functions $f$ and $g$ as follows.
\begin{definition} \label{def:wLDOS}
Let $H$ be the Hamiltonian of a finite-dimensional tight-binding model with $d = 1$, and let $X$ be the position operator \eqref{eq:pos_R1}. Let $E$ and $x$ be real numbers, let $f(\xi)$ and $g(\xi)$ be positive $L^1$ functions such that $g$ is compactly supported with $0 \leq g(\xi) \leq 1$, and let $F_{E,x}(H,X)$ be as in \eqref{eq:mLDOS_1d}. The windowed local density of states (wLDOS) at energy $E$ and position $x$ with window functions $f$ and $g$ is then defined by
\begin{equation} \label{eq:mLDOS}
    W_x(E) := \Tr F_{E, x}(H, X).
  \end{equation}
\end{definition}
Note that the wLDOS can be defined at any $x \in \field{R}$, even if $x$ is not the position of a site.

An alternative equivalent formulation which will be useful is as follows. For arbitrary Hermitian matrices $A$, we have the identity
\begin{equation}
    \Tr ( A^2 ) = \| A \|^2_F
\end{equation}
where $\| A \|_F$ denotes the Frobenius norm of a matrix $A = (a_{ij})_{1 \leq i,j \leq \mathcal{N}}$
\begin{equation}
    \| A \|_F = \sqrt{ \sum_{i,j = 1}^{\mathcal{N}} |a_{ij}|^2 }.
\end{equation}
Since $F_{E,x}(H,X)$ is positive by construction, we have 
\begin{equation} \label{eq:trace_frob}
    \Tr F_{E,x}(H,X) = \left\| F^{\frac{1}{2}}_{E,x}(H,X) \right\|_F^2
\end{equation}
so that the right-hand side of \eqref{eq:trace_frob} gives an alternative definition of the wLDOS.

We now claim the following proposition, which establishes that the
wLDOS defined by Definition \ref{def:wLDOS}, with a particular
  choice of window functions, reduces to the standard LDOS for
finite-dimensional tight-binding models. 
\begin{proposition} \label{prop:wLDOS_reduces}
Let $f$ be a positive $L^1$ function. Let $\left\{ g_n \right\}_{n \in \mathcal{I}}$ and $\left\{ x_n \right\}_{n \in \mathcal{I}}$ denote sets of compactly supported functions each satisfying $0 \leq g_n \leq 1$, and real numbers respectively. For each $n \in \mathcal{I}$, let $W_{x_n}(E)$ denote the wLDOS at energy $E$ and position $x_n$ using window functions $f$ and $g_n$. Then:
\begin{enumerate}
\item If the functions $g_n$ centered at $x_n$ form a partition of
  unity 
\begin{equation} \label{eq:part_of_unity}
    \sum_{n \in \mathcal{I}} g_n(\xi - x_n) = 1,
\end{equation}
then
\begin{equation} \label{eq:add_to_wDOS}
    \frac{1}{\mathcal{N}} \sum_{n \in \mathcal{I}} W_{x_n}(E) = W(E)
\end{equation}
where $W(E)$ is the wDOS defined by \eqref{eq:wDOS}.
\item If each point $x_n$ is chosen as the co-ordinate of the $n$th site (so that $\mathcal{I} = \left\{1,...,\mathcal{N}\right\}$), the functions $g_n$ are chosen such that exactly one site is in the support of $g_n(\xi - x_n)$ for all $n \in \mathcal{I}$, and $g_n(0) = 1$, then \eqref{eq:add_to_wDOS} holds and for each $n \in \mathcal{I}$ 
\begin{equation}
    W_{x_n}(E) = \sum_{m = 1}^M \bra{ \delta_n^m } f(H-E) \ket{ \delta_n^m }
\end{equation}
which is exactly the windowed (in energy) LDOS \eqref{eq:wieLDOS}.
\end{enumerate}
\end{proposition}
As an example of a set of functions satisfying \eqref{eq:part_of_unity}, note that the bump function \eqref{eq:f} satisfies
\begin{equation} \label{eq:g_part_of_unity}
    \sum_{n \in \field{Z}} g(\xi - 2 n) = 1.
\end{equation}
\begin{proof}[Proof of Proposition \ref{prop:wLDOS_reduces}]
Using the cyclic property of the trace and the definition \eqref{eq:mLDOS_1d}, we have that
\begin{equation} \label{eq:trace}
    \Tr \left( F_{E,x}(H,X) \right) = \Tr \bigl( g(X - x) f(H - E) \bigr).
\end{equation}
Now let $\left\{ g_n(\xi) \right\}_{n \in \mathcal{I}}$ and $\left\{ x_n \right\}_{n \in \mathcal{I}}$ be as in Proposition \ref{prop:wLDOS_reduces}, and assume \eqref{eq:part_of_unity}. Then
\begin{equation}
    \frac{1}{\mathcal{N}} \sum_{n \in \mathcal{I}} W_{x_n}(E) = \frac{1}{\mathcal{N}} \Tr \left( \sum_{n \in \mathcal{I}} g(X-x_n) f(H-E) \right) = \frac{1}{\mathcal{N}} \Tr f(H - E),
\end{equation}
which is nothing but \eqref{eq:wDOS}, so (1) is proved. For (2), note that we can expand the trace in \eqref{eq:trace} in the basis of eigenvectors of $X$ to derive
\begin{equation} \label{eq:W_xE}
\begin{split}
    W_x(E) = \Tr \bigl( g(X - x) f(H - E) \bigr) = \sum_{n = 1}^N \sum_{m = 1}^M \ip{ g(X - x) \delta^m_n }{ f(H - E) \delta^m_n }.
\end{split}
\end{equation}
Using the spectral representation of $X$ we have that
\begin{equation} \label{eq:W_xE_2}
    W_x(E) = \sum_{n = 1}^N \sum_{m = 1}^M g(x_n - x) \bra{ \delta^m_n } f(H - E) \ket{ \delta^m_n }.
\end{equation}
Part (2) of Proposition \ref{prop:wLDOS_reduces} is now clear from the assumptions that exactly one site is in the support of each $g(x_n - x)$ and $g(0) = 1$.
\end{proof}


We now define the wLDOS in higher dimensions. In $d$ dimensions, recall that we define position operators by
\begin{equation} \label{eq:pos_Rd}
    X_l := \sum_{n = 1}^N x_{n,l} \sum_{m = 1}^M \ket{\delta^m_n} \bra{\delta^m_n} \quad 1 \leq l \leq d
\end{equation}
where $\vec{x}_n = (x_{n,1},x_{n,2},...,x_{n,d})$ are the co-ordinates of the $n$th site. We will use the obvious notation $\vec{X} := (X_1,...,X_d)$. Let $f(\xi)$ be a positive $L^1$ function and $g(\vec{\xi})$ denote a compactly supported function $\field{R}^d \rightarrow \field{R}$ such that $0 \leq g(\vec{\xi}) \leq 1$. Then for arbitrary real $E$ and $\vec{x} = (x_1,...,x_d) \in \field{R}^d$, let
\begin{equation} \label{eq:mLDOS_Rd}
    F_{E,\vec{x}}(H,\vec{X}) := g^{\frac{1}{2}}(\vec{X}-\vec{x}) f(H - E) g^{\frac{1}{2}}(\vec{X}-\vec{x}).
\end{equation}

\begin{definition} \label{def:wLDOS_Rd}
Let $H$ be the Hamiltonian of a finite-dimensional tight-binding model, and let $X_l, 1 \leq l \leq d$ denote the position operators \eqref{eq:pos_Rd}. Let $E$ be a real number and $\vec{x} = (x_1,...,x_d) \in \field{R}^d$, let $f(\xi)$ be a positive $L^1$ function and $g(\vec{\xi}) \in L^1(\field{R}^d)$ be compactly supported with $0 \leq g(\vec{\xi}) \leq 1$, and let $F_{E,\vec{x}}(H,X)$ be as in \eqref{eq:mLDOS_Rd}. The windowed local density of states (wLDOS) at energy $E$ and position $\vec{x}$ with window functions $f$ and $g$ is then defined by
\begin{equation} \label{eq:wLDOS_Rd}
  W_{\vec{x}}(E) := \Tr F_{E, \vec{x}}(H, \vec{X}).
\end{equation}
\end{definition}

Since the proof is identical to that of Proposition \ref{prop:wLDOS_reduces}, we state the following without proof.
\begin{proposition} \label{prop:wLDOS_reduces_Rd}
Let $\left\{ g_n(\vec{\xi}) \right\}_{n \in \mathcal{I}}$ and $\left\{ \vec{x}_n \right\}_{n \in \mathcal{I}}$ denote sets of compactly supported functions $\field{R}^d \rightarrow \field{R}$ each satisfying $0 \leq g_n \leq 1$, and real numbers respectively. Let $f(\xi)$ be a positive $L^1$ function. For each $n \in \mathcal{I}$, let $W_{\vec{x}_n}(E)$ denote the wLDOS at energy $E$ and position $x_n$ using window functions $f$ and $g_n$. Then:
\begin{enumerate}
\item If the functions $g_n$ centered at $\vec{x}_n$ form a partition of unity
\begin{equation} \label{eq:part_of_unity_Rd}
    \sum_{n \in \mathcal{I}} g_n(\vec{\xi} - \vec{x}_n) = 1,
\end{equation}
then
\begin{equation} \label{eq:add_to_wDOS_Rd}
    \frac{1}{\mathcal{N}} \sum_{n \in \mathcal{I}} W_{\vec{x}_n}(E) = W(E).
\end{equation}
\item If each point $\vec{x}_n$ is chosen as the co-ordinate of the $n$th site (so that $\mathcal{I} = \left\{1,...,\mathcal{N}\right\}$) and the functions $g_n$ are chosen such that exactly one site is in the support of $g_n(\vec{\xi} - \vec{x}_n)$ for all $n \in \mathcal{I}$, then \eqref{eq:add_to_wDOS} holds, and the wLDOS reduces to the windowed (in energy) LDOS:
\begin{equation}
    W_{\vec{x}_n}(E) = \sum_{m = 1}^M \bra{ \delta^m_n } f(H - E) \ket{\delta_n^m}.
  \end{equation}
\end{enumerate}
\end{proposition}
As an example of a set of functions satisfying \eqref{eq:part_of_unity_Rd} we can take products and translates of the bump function \eqref{eq:f}. For example in dimension $d = 2$ we have
\begin{equation}
    \sum_{(n_1,n_2) \in \field{Z}^2} g(\xi_1 - 2 n_1) g(\xi_2 - 2 n_2) = 1.
\end{equation}

\begin{remark}
We briefly note some practical considerations which should be taken into account when numerically computing the wLDOS. First, note that a na\"ive computation of $f(H - E)$ would require a potentially expensive diagonalization of $H$. Assuming $f(\xi)$ is sufficiently smooth this can be avoided by approximating $f(\xi)$ in $L^\infty$ by a polynomial $f_p(\xi)$ so that $f_p(H-E)$ can be accurately computed by merely repeatedly applying $H-E$. This is known as the kernel polynomial method \cite{Weisse2006,lin2016approximating} and can be rigorously justified using Bernstein's theorem. Second, note that when $F_{E,\vec{x}}(H,\vec{X})$ has large rank it may be costly to directly evaluate the trace defining the wLDOS. In this case it may be preferable to compute the trace using a randomized algorithm \cite{doi:10.1080/03610918908812806,lin2016approximating}.
\end{remark}

\begin{remark}
We have seen that for finite systems the DOS can be recovered from the wLDOS by averaging over sites. Assuming the wLDOS can be computed efficiently, one can imagine a scheme for efficiently computing the DOS of a large system by averaging over local computations of the wLDOS. The efficiency of this approach would come from the fact that each local computation could be computed independently and hence could be parallelized.
\end{remark}







\section{Proof that the wLDOS is local} \label{sec:local_proof}

We now prove the wLDOS is local in the sense that it can be computed with a finite truncation of the system Hamiltonian nearby to the point of interest up to error which can be made arbitrarily small. We will prove this initially for finite-dimensional models without aiming for optimal constants. In the next section we will introduce natural assumptions which will significantly improve these constants and allow us to exploit locality to define wLDOS for a class of infinite-dimensional tight-binding models. We will state and prove the result in one spatial dimension for clarity and then state the general result for models in $\field{R}^d$ without proof since the proof is similar to the one dimensional case.

Let $f(\xi)$ and $g(\xi)$ be as in Section \ref{sec:wLDOS_sec}, i.e., $f$ is a positive $L^1$ function, and $g \in L^1$ compactly supported with $0 \leq g \leq 1$. We now additionally assume that $f$ is sufficiently smooth (specifically, we assume \eqref{eq:f_regularity}).
Now let $k(\xi) \in L^1$ be
compactly supported with $0 \leq k \leq 1$ equalling $1$ on the
support of $g(\xi)$. For example, if $g$ is given by \eqref{eq:f}, we
can take $k$ to be
\begin{equation}
    k(\xi) = \begin{cases} 0 & \xi \leq -4 \\ \frac{1}{2} (\xi + 4)^2 & -4 < \xi \leq -3 \\ 1 - \frac{1}{2} (\xi + 2)^2 & -3 < \xi \leq -2 \\ 1 & -2 < \xi \leq 2 \\ 1 - \frac{1}{2} (\xi - 2)^2 & 2 < \xi \leq 3 \\ \frac{1}{2} (\xi - 4)^2 & 3 < \xi \leq 4 \\ 0 & 4 \leq \xi. \end{cases}
\end{equation}
We think of $k$ as modeling spatial truncation of the Hamiltonian. 

We first give an outline of our results before stating theorems. Detailed proofs will be postponed to an Appendix. Recall the definition of the wLDOS (Definition \ref{def:wLDOS}). The statement that the wLDOS is local is then the statement that for any real $E$ and $x$,
\begin{equation} \label{eq:to_prove_1}
\begin{split}
    W_x(E) &= \Tr F_{E,x}(H,X) \\ 
    &\approx \Tr g^{\frac{1}{2}}(X - x) f\left( k(X - x) (H-E) k(X - x) \right) g^{\frac{1}{2}}(X - x).
\end{split}
\end{equation}
Note that the right-hand side only involves the ``spatially
  truncated'' Hamiltonian
  $k(X - x) (H - E) k(X - x)$. The main
step to prove \eqref{eq:to_prove_1} will be an estimate
\begin{equation} \label{eq:to_prove}
\begin{split}
    &\left\| g^{\frac{1}{2}}(X-x)f(H-E)g^{\frac{1}{2}}(X-x) \right\|     \\
    &\approx \left\| g^{\frac{1}{2}}(X-x)f\bigl(k(X-x)(H-E)k(X-x)\bigr)g^{\frac{1}{2}}(X-x) \right\| \\
\end{split}
\end{equation}
in the \emph{operator} norm. Since \eqref{eq:to_prove} is equivalent to a statement about Frobenius norms using \eqref{eq:trace_frob}, we can pass to the estimate \eqref{eq:to_prove_1} using equivalence of finite-dimensional norms (for large system sizes this step will give a large constant which can be avoided by making natural assumptions on $H$, see Section \ref{sec:infinite_systems}).

We now move to stating our results rigorously. We will establish \eqref{eq:to_prove} in two steps. The first and more difficult step is to prove the following lemma.
\begin{lemma} \label{lem:locality_opnorm}
Suppose $H$ and $X$ are finite-dimensional Hermitian operators. Let $g$ and $k$ be positive $L^1$ functions such that 
$0\leq g\leq1$, $0\leq k\leq1$, and $kg=g$. Let $f$ be a positive $L^1$ function such that 
\begin{equation} \label{eq:f_regularity}
    \inty{-\infty}{\infty}{ \left( 1 + |t|^2 \right) \left| \widehat{f}(t) \right| }{t} < \infty, \quad \widehat{f}(t) := \frac{1}{2 \pi} \inty{-\infty}{\infty}{ e^{- i t \xi} f(\xi) }{\xi}.
\end{equation}
Then
\begin{equation} \label{eq:locality_bound_toprove}
\left\Vert g^{\frac{1}{2}}(X)f(H)g^{\frac{1}{2}}(X)-g^{\frac{1}{2}}(X)f\bigl(k(X)Hk(X)\bigr)g^{\frac{1}{2}}(X)\right\Vert \leq C_1\left\Vert \left[k(X),H\right]\right\Vert
\end{equation}
where
\begin{equation} \label{eq:C}
    C_1 = \int_{-\infty}^{\infty} |t| ( 1 + |t| \| H \|) \left| \widehat{f}(t)\right|\,dt.
  \end{equation}
\end{lemma}
Condition \eqref{eq:f_regularity} holds as long as $f$ is twice differentiable with the Fourier transform of $f''$ in $L^1$.

The second step is to prove that the right-hand side of \eqref{eq:locality_bound_toprove} can be made arbitrarily small by an appropriate choice of $k(\xi)$. Specifically, we will prove the following proposition. 
\begin{proposition} \label{prop:k_alpha}
Let $g(\xi) \in L^1$ satisfy $0 \leq g \leq 1$ and have support confined to the interval $[-L,L]$ for some fixed $L > 0$. Then it is possible to construct compactly supported functions $k_\alpha(\xi)$ defined for each $\alpha > 0$ which equal $1$ for all $\xi \in [-L,L]$ and such that
\begin{equation} \label{eq:kXH}
    \| [k_\alpha(X),H] \| \leq C_2 \alpha \| [X,H] \|,
\end{equation}
where $C_2 > 0$ is a constant independent of $\alpha$. The support of $k_\alpha(\xi)$ is confined to the interval $\left[-\frac{L+4}{\alpha},\frac{L+4}{\alpha}\right]$.
\end{proposition}
For the proofs of Lemma \ref{lem:locality_opnorm} and \ref{prop:k_alpha}, see Appendix \ref{app:locality_lemmas}.

Combining Lemma \ref{lem:locality_opnorm} with Proposition \ref{prop:k_alpha} we have the following theorem which makes \eqref{eq:to_prove_1} rigorous, establishing that $W_x(E)$ can be computed using the spatially truncated Hamiltonian $k_\alpha(X-x) H k_\alpha(X-x)$ up to error of order $\alpha$ for any $\alpha > 0$.
\begin{theorem} \label{th:first_locality_thm}
Let $H$ be a finite-dimensional Hermitian operator. Let $E$ and $x$ be real numbers. Let $W_x(E)$ be the wLDOS defined by Definition \ref{def:wLDOS} with window functions $f \in L^1$, positive and satisfying \eqref{eq:f_regularity}, and $g \in L^1$ compactly supported with $0 \leq g \leq 1$. Let $k_\alpha(X)$ for each $\alpha > 0$ be the functions constructed in Proposition \ref{prop:k_alpha}. Then
\begin{equation} \label{eq:main_bd}
    \left| W_x(E) - \Tr \left( g^{\frac{1}{2}}(X-x) f\left( k_\alpha(X - x) (H - E) k_\alpha(X - x) \right) g^{\frac{1}{2}}(X - x) \right) \right| \leq C \alpha,
\end{equation}
where $C > 0$ is a constant independent of $\alpha$.
\end{theorem}
\begin{proof}
Combining Lemma \ref{lem:locality_opnorm} with Proposition \ref{prop:k_alpha} gives a bound in the operator norm depending only on $\| H \|$ and $\| [H,X] \|$, both of which are finite since we work in a finite-dimensional space. To pass to the estimate in the Frobenius norm we use equivalence of finite-dimensional norms in the space $\field{C}^{\mathcal{N}}$.
\end{proof}
Note that the proof uses very na\"ive estimates that the resulting
constant $C$ will grow with system size. In the next section we will
introduce assumptions that allow for estimates which are uniform in
the system size.

For completeness we state the $d$-dimensional result without proof.
\begin{theorem} \label{th:first_locality_thm_ddim} Let $H$ be a
  finite-dimensional Hermitian operator. Let $E$ be real and
  $\vec{x} \in \field{R}^d$. Let $W_{\vec{x}}(E)$ be the wLDOS defined
  by Definition \ref{def:wLDOS_Rd} with window functions $f \in L^1$,
  positive and with Fourier transform satisfying
  \eqref{eq:f_regularity}, and $g \in L^1(\field{R}^d)$ compactly
  supported with $0 \leq g \leq 1$. Let $k_\alpha(\vec{X})$ for each
  $\alpha > 0$ be a tensor product of the the one-dimensional
  functions constructed in Proposition \ref{prop:k_alpha}. Then
  \begin{equation} \label{eq:main_bd}
    \left| W_{\vec{x}}(E) - \Tr \left( g^{\frac{1}{2}}(\vec{X}-\vec{x}) f\left( k_\alpha(\vec{X} - \vec{x}) (H - E) k_\alpha(\vec{X} - \vec{x}) \right) g^{\frac{1}{2}}(\vec{X} - \vec{x}) \right) \right| \leq C \alpha,
  \end{equation}
  where $C > 0$ is a constant independent of $\alpha$.
\end{theorem}

\section{Defining and computing the wLDOS of infinite systems using locality} \label{sec:infinite_systems}

In this section we will introduce natural assumptions which will allow for locality estimates which are uniform in system size. Using these locality estimates we will then show that the wLDOS is well-defined and computable for a broad class of infinite-dimensional tight-binding systems. 

\subsection{Approximation of infinite-dimensional tight-binding models by finite-dimensional tight-binding models}

Consider a tight-binding model on an infinite number of sites. As examples, we can consider an electron hopping on an infinite periodic lattice, or on an infinite quasicrystal lattice, or on random perturbations of such lattices. We will take the model Hilbert space to be $\mathcal{H} := \ell^2(V) \otimes \field{C}^M$, where $V$ denotes the set of sites and $M$ denotes the number of internal degrees of freedom, and denote the model Hamiltonian, a self-adjoint operator $\mathcal{H} \rightarrow \mathcal{H}$, by $H_{\infty}$. We define position operators by
\begin{equation} \label{eq:pos_Rd_inf}
    X_l := \sum_{n \in V} x_{n,l} \sum_{m = 1}^M \ket{\delta^m_n} \bra{\delta^m_n} \quad 1 \leq l \leq d
\end{equation}
where $\vec{x}_n = (x_{n,1},x_{n,2},...,x_{n,d})$ denotes the co-ordinates of each site.

Suppose we fix a real number $E$ and $\vec{x} \in \field{R}^d$. For any $R > 0$, we can define a finite-dimensional tight-binding model with Hamiltonian $H_R$ by restricting the infinite model to sites in the ball of radius $R$ about $\vec{x}$, and compute the wLDOS of this model at $E$ and $\vec{x}$. Since we recover the infinite-dimensional model in the limit $R \rightarrow \infty$, it is natural to ask whether it makes sense to take the limit of the wLDOS of the sequence of finite models and define the wLDOS of the infinite model by this limit. In the previous section we proved that the wLDOS can be computed from a truncation of the Hamiltonian to a region nearby the point of interest. Formally then, the sequence of wLDOS values should converge as $R \rightarrow \infty$. To make this rigorous, we have to control the constant in the estimate \eqref{eq:main_bd} as a function of $R$. The proof of Theorem \ref{th:first_locality_thm} clearly does not provide this since, for example, we invoke equivalence of norms in $\field{C}^{\mathcal{N}}$, where $\mathcal{N}$ will increase as function of $R$.
To pass to the limit $R \rightarrow \infty$, we require three natural assumptions which we expect will be satisfied by any physically reasonable tight-binding model. We first assume that the Hamiltonian $H_\infty$ is \emph{local} in the following sense. 
\begin{assumption} \label{as:locality}
Let $X_l$, $1 \leq l \leq d$ denote the position operators defined by \eqref{eq:pos_Rd_inf}, extended to all of $\field{R}^d$. We assume that $H_{\infty}$ is local in the sense that there exists a constant $C_{loc} > 0$ such that
\begin{equation}
    \sup_{1 \leq l \leq d} \bigl\lVert [ X_l , H_{\infty} ] \bigr\rVert \leq C_{loc}.
\end{equation}
\end{assumption}
Assumption \ref{as:locality} can be roughly stated as ``$H_\infty$ is a narrowly banded matrix in the position basis''. We next assume that $H_\infty$ is \emph{bounded}.
\begin{assumption} \label{as:bounded}
We assume that $H_\infty$ is bounded in the sense that there exists a constant $C_{norm} > 0$ such that
\begin{equation}
    \| H_\infty \| \leq C_{norm}.
\end{equation}
\end{assumption}
The final assumption rules out some pathological situations where e.g. balls with finite radius can contain an unbounded number of sites.
\begin{assumption} \label{as:finite_rank}
Let $g(\xi) \in L^1(\field{R}^d)$ be compactly supported with $0 \leq g \leq 1$. Then we assume that the rank of the matrix
\begin{equation}
    g^{\frac{1}{2}}(\vec{X} - \vec{x}) H_\infty g^{\frac{1}{2}}(\vec{X} - \vec{x})
\end{equation}
is uniformly bounded above for all $\vec{x} \in \field{R}^d$ by a positive integer $M_{upper}$.
\end{assumption}

Note that Assumptions \ref{as:locality}, \ref{as:bounded}, and \ref{as:finite_rank} are trivial for Hamiltonians $H$ of fixed finite-dimensional tight-binding models. 

We now have the following.
\begin{theorem} \label{th:second_locality_thm}
Let $H_\infty$ be an infinite-dimensional tight-binding Hamiltonian satisfying Assumptions \ref{as:locality}, \ref{as:bounded}, and \ref{as:finite_rank}. Let $E$ be real, and $\vec{x} \in \field{R}^d$. For any $R > 0$, define $H_R$ as the tight-binding Hamiltonian obtained by truncating $H_\infty$ to the set of sites within a ball of radius $R$ about the point $E$ and $\vec{x}$. Let $f \in L^1$ be positive and $g \in L^1$ be compactly supported with $0 \leq g \leq 1$, and let $W_{\vec{x},R}(E)$ be the wLDOS defined by Definition \ref{def:wLDOS} for the Hamiltonian $H_R$. Then:
\begin{enumerate}
\item The limit $W_{\vec{x},\infty}(E) := \lim_{R \rightarrow \infty} W_{\vec{x},R}(E)$ exists and equals
\begin{equation} \label{eq:limit}
    W_{\vec{x},\infty}(E) = \Tr  \left( g^{\frac{1}{2}}(\vec{X}-\vec{x}) f(H_\infty - E) g^{\frac{1}{2}}(\vec{X} - \vec{x}) \right).
\end{equation}
\item The limit $W_{\vec{x},\infty}(E)$ can be computed by the formula 
\begin{equation} \label{eq:formula}
    W_{\vec{x},\infty}(E) = \Tr \left( g^{\frac{1}{2}}(\vec{X} - \vec{x}) f\left( k_{\alpha}(\vec{X}-\vec{x}) (H_R-E) k_\alpha(\vec{X}-\vec{x}) \right) g^{\frac{1}{2}}(\vec{X}-\vec{x}) \right) + O(\alpha),
\end{equation}
as long as $R > 0$ is sufficiently large that 
\begin{equation}
    k_\alpha ( \vec{X} - \vec{x} ) (H_R - E) k_\alpha(\vec{X} - \vec{x}) = k_\alpha ( \vec{X} - \vec{x} ) (H_{\infty} - E) k_\alpha ( \vec{X} - \vec{x} ).
\end{equation}
\end{enumerate}
\end{theorem}
\begin{proof}
Let $R > 0$ be arbitrary. Then applying Lemma \ref{lem:locality_opnorm} and Proposition \ref{prop:k_alpha} to the truncated model with radius $R$ we have a bound of the form \eqref{eq:main_bd} in the operator norm with a constant depending on $\| [X_l,H_R] \|, 1 \leq l \leq d$ and $\| H_R \|$. Under Assumptions \ref{as:locality} and \ref{as:bounded}, these can both be bounded independent of $R$. To pass to the bound \eqref{eq:main_bd} in the Frobenius norm, we invoke equivalence of finite-dimensional norms in the space $\field{C}^{M_{upper}}$ using Assumption \ref{as:finite_rank}. We now have an estimate
\begin{equation} \label{eq:first}
    \left| W_{x,R}(E) - \Tr \left( g^{\frac{1}{2}}(\vec{X} - \vec{x}) f\left( k^{\frac{1}{2}}_{\alpha}(\vec{X}-\vec{x}) (H_R-E) k_\alpha(\vec{X}-\vec{x}) \right) g^{\frac{1}{2}}(\vec{X}-\vec{x}) \right) \right| \leq C \alpha
\end{equation}
where $C > 0$ is independent of both $R$ and $\alpha$. Since $k_{\alpha}(\vec{X}-\vec{x})$ is a cutoff, by taking $R$ sufficiently large we can ensure that
\begin{equation} \label{eq:cutoff_H}
    k_{\alpha}(\vec{X}-\vec{x}) (H_R-E) k_\alpha(\vec{X}-\vec{x}) = k_{\alpha}(\vec{X}-\vec{x}) (H_\infty-E) k_\alpha(\vec{X}-\vec{x}).
\end{equation}
We now claim that the sequence $\{ W_{\vec{x},R}(E) \}$ is Cauchy. Let $\epsilon > 0$ be arbitrary. For arbitrary $R, R'$, \eqref{eq:first} implies that 
\begin{equation} \label{eq:here}
\begin{split}
    &\left| W_{\vec{x},R}(E) - W_{\vec{x},R'}(E) \right| \\
    &= \left| \Tr \left( g^{\frac{1}{2}}(\vec{X} - \vec{x}) f\left( k_{\alpha}(\vec{X}-\vec{x}) (H_R-E) k_\alpha(\vec{X}-\vec{x}) \right) g^{\frac{1}{2}}(\vec{X}-\vec{x}) \right) \right. \\
    &\left. \quad - \Tr \left( g^{\frac{1}{2}}(\vec{X} - \vec{x}) f\left( k_{\alpha}(\vec{X}-\vec{x}) (H_{R'}-E) k_\alpha(\vec{X}-\vec{x}) \right) g^{\frac{1}{2}}(\vec{X}-\vec{x}) \right)^{\frac{1}{2}} \right| + O(\alpha).
\end{split}
\end{equation}
Now take $\alpha$ small enough such that the $O(\alpha)$ term is $< \epsilon$, and then take $R$ and $R'$ sufficiently large that \eqref{eq:cutoff_H} holds for both terms so that the other term in \eqref{eq:here} vanishes.

To see \eqref{eq:limit}, note that for arbitrary $R$ and $\alpha$,
\begin{equation}
    W_{\vec{x},R}(E) = \Tr \left( g^{\frac{1}{2}}(\vec{X}-\vec{x}) f\left( k_{\alpha}(\vec{X}-\vec{x}) (H_R - E) k_{\alpha}(\vec{X}-\vec{x}) \right) g^{\frac{1}{2}}(\vec{X} - \vec{x}) \right) + O(\alpha).
\end{equation}
Taking the limit $R \rightarrow \infty$ on both sides we have
\begin{equation} \label{eq:intermediate}
    W_{\vec{x},\infty}(E) = \Tr \left( g^{\frac{1}{2}}(\vec{X}-\vec{x}) f\left( k_{\alpha}(\vec{X}-\vec{x}) (H_\infty - E) k_{\alpha}(\vec{X}-\vec{x}) \right) g^{\frac{1}{2}}(\vec{X} - \vec{x}) \right) + O(\alpha).
\end{equation}
Taking the limit $\alpha \rightarrow 0$ (note that $\lim_{\alpha \rightarrow 0} k_{\alpha}(\vec{X}-\vec{x}) = 1$) now implies \eqref{eq:limit}. To see \eqref{eq:formula}, fix $\alpha > 0$ in \eqref{eq:intermediate}. Using \eqref{eq:cutoff_H} we have that for sufficiently large $R$ (depending on $\alpha$) that the right-hand side equals \eqref{eq:formula}.
\end{proof}
We can now make the following definition.
\begin{definition}
Let $H_\infty$ be the Hamiltonian of a tight-binding model on an infinite lattice satisfying Assumptions \ref{as:locality}, \ref{as:bounded}, and \ref{as:finite_rank} and let $X_l, 1 \leq l \leq d$ denote the position operators \eqref{eq:pos} extended to $\field{R}^d$. Let $E$ be real and $\vec{x} \in \field{R}^d$, let $f(\xi)$ be a positive $L^1$ function such that \eqref{eq:f_regularity} holds, and $g(\vec{\xi}) \in L^1(\field{R}^d)$ be compactly supported with $0 \leq g(\vec{\xi}) \leq 1$. We define the windowed local density of states (wLDOS) at energy $E$ and position $\vec{x}$ with window functions $f$ and $g$ by
\begin{equation}
\begin{split}
    W_{\vec{x}}(E) &:= \Tr F_{E,\vec{x}}(H_\infty,\vec{X})   \\
    F_{E,\vec{x}}(H_\infty,\vec{X}) &:= g^{\frac{1}{2}}(\vec{X} - \vec{x}) f(H_\infty-E) g^{\frac{1}{2}}(\vec{X} - \vec{x}).
\end{split}
\end{equation}
\end{definition}
That this quantity is well-defined and computable is guaranteed by Theorem \ref{th:second_locality_thm}.

\section{The Fibonacci SSH model} \label{sec:Fib_SSH}

In this section we introduce a one-dimensional model system which we refer to as the Fibonacci SSH model. We choose to study this model because it lacks any translational symmetry, making the wLDOS an important tool for understanding the electronic states of the system.

Before we can define the Fibonacci SSH model, we must first review the Fibonacci quasicrystal construction.

\subsection{The Fibonacci Quasicrystal} \label{sec:Fib_quasicrystal}

The Fibonacci quasicrystal is a one dimensional chain made of two sorts of ``links'' which have lengths $S$ and $L$ with $S < L$. We take the lengths in a fixed ratio of the golden mean 
\begin{equation}
\frac{|L|}{|S|}=\phi=\frac{1+\sqrt{5}}{2}.
\end{equation}
We can form an infinite quasiperiodic chain by starting with a series of links and then repeatedly applying the replacement rules 
\begin{equation}
\begin{array}{l}
    S \mapsto L,\\
    L \mapsto LS.
\end{array}
\end{equation}
If we start with the single letter $S$, we obtain the sequence
\begin{equation}
\begin{array}{l}
  S\\
  L\\
  LS\\
  LSL\\
  LSLLS\\
  LSLLSLSL\\
  LSLLSLSLLSLLS
\end{array}
\end{equation}
and so on. It will be more convenient to work with a sequence which grows in two directions rather than one. To obtain such a sequence, instead of starting with $S$ we start with $LL$. Applying the replacement rules we find
\begin{equation} \label{eq:non-per}
\begin{array}{c}
L\!\!\centerdot\!\!L\\
LS\!\!\centerdot\!\!LS\\
LSL\!\!\centerdot\!\!LSL\\
LSLLS\!\!\centerdot\!\!LSLLS\\
LSLLSLSL\!\!\centerdot\!\!LSLLSLSL\\
\end{array}
\end{equation}
and so on, where $\centerdot$ indicates the center of the
sequence. 
We will shortly want to identify the ends of each stage of the
sequence. When we do this, we would like to guarantee that we do not
create sequences of letters, known as words, which did not appear in
the original sequence. We call such words invalid words.

To obtain a sequence such that identifying the ends of each stage does not create invalid words, we start with $LLS$ instead of $LL$. In this case, invalid words are not created by identifying ends because $LLSL$ and $LSLL$ and $SLLS$ all are valid words, appearing by stage 3. The first five sequences obtained from the quasicrystal construction are then
\begin{equation} \label{eq:all_stages}
\begin{split}
&\text{stage 1: } \\
&\text{stage 2: } \\
&\text{stage 3: } \\
&\text{stage 4: } \\
&\text{stage 5: } \\
\end{split}
\begin{split}
L & \!\!\centerdot\!\!LS\\
LS & \!\!\centerdot\!\!LSL\\
LSL & \!\!\centerdot\!\!LSLLS\\
LSLLS & \!\!\centerdot\!\!LSLLSLSL\\
LSLLSLSL & \!\!\centerdot\!\!LSLLSLSLLSLLS.
\end{split}
\end{equation}
Notice that stage 1 appears at the center of stage 3, stage 3 appears at the center of stage 5, and so on. We can therefore consider the sequence of odd stages as a sequence of chains which grow at their ends. The infinite quasicrystal is defined as the infinite limit of this sequence. Note that if we took instead the even stages, we would have something locally indistinguishable from what we are using (see Example 4.6 of \cite{baake_grimm_2013}).


Our final task in this section is to compute values for the lengths of $S$ and $L$ so that the average distance between vertices (points between links) is $1$. The number of symbols at stage $n$ is a Fibonacci number\footnote{The Fibonacci numbers are defined by $F_0 = 0$, $F_1 = 1$ and $F_n = F_{n-1} + F_{n-2}$ for all $n > 1$. A straightforward induction proves the stronger statement that the numbers of copies of $S$ and $L$ at stage $n$ are $F_{n+1}$ and $F_{n+2}$ respectively.},
\begin{equation}
    F_{n+3}\approx\frac{1}{\sqrt{5}}\phi^{n+3}.
\end{equation}
Suppose we replace $S$ by an edge of length $1$ and $L$ by an edge
of length $\phi$ where 
\begin{equation}
    \phi=\frac{1+\sqrt{5}}{2}.
\end{equation}
(Thus $\phi^{2}=\phi+1$, $\phi^{3}=2\phi+1$, etc.). Since the total
length represented by $L$ is $\phi$ and the total length represented by $LS$ is $\phi+1$, the length of the resulting finite quasilattice will be growing at
each stage by the factor $\phi$. At stage $1$ the total length of
the quasilattice is $2\phi+1$ so the total length of the quasilattice
at stage $n$ is $\phi^{n+2}$. This means the average distance between vertices is
\begin{equation}
\frac{\phi^{n+2}}{F_{n+3}}\approx\frac{\sqrt{5}\phi^{n+2}}{\phi^{n+3}}=\frac{\sqrt{5}}{\phi}.
\end{equation}
To get the average distance between vertices to be $1$, we rescale,
and so use
\begin{equation}
|S|=\frac{\phi}{\sqrt{5}}=\frac{\phi+2}{5}\approx0.7236
\end{equation}
\begin{equation}
|L|=\frac{\phi^{2}}{\sqrt{5}}=\frac{3\phi+1}{5}\approx1.1708
\end{equation}

\subsection{The Fibonacci SSH model} \label{sec:Fib_SSH_2}

The original SSH model \cite{Su1979} describes hopping on a one-dimensional lattice where the hopping amplitudes alternate between two values, known as the inter-site and onsite hopping amplitudes. In the Fibonacci SSH model, the onsite hopping amplitude is held fixed while the inter-site hopping amplitude takes on one of two values, with the choice determined by the Fibonacci quasicrystal construction constructed in the previous section.

The model is defined on the Hilbert space $\ell^{2}(V)\otimes\mathbb{C}^{2}$ with $V$ the set of vertices of the infinite Fibonacci quasicrystal. As usual we define the one-dimensional position operator by
\begin{equation}
    X = \sum_{n \in V} x_n \sum_{m = 1}^2 \ket{\delta_n^m} \bra{\delta_n^m}
\end{equation}
where $x_n$ is the co-ordinate of each vertex. Let $\sigma_x$ denote the Pauli matrix
\begin{equation}
    \sigma_x = \begin{pmatrix} 0 & 1 \\ 1 & 0 \end{pmatrix}.
\end{equation}
The Hamiltonian is
\begin{equation}
    H = \sum_{n \in V} t_{o} \ket{\delta_n^2} \bra{\delta_n^1} + t_{i}(n,n+1) \ket{\delta_{n+1}^1} \bra{\delta_n^2} + h.c.
\end{equation}
where $h.c.$ denotes the Hermitian conjugate. Here $t_{o}$ is a real constant defining the onsite hopping amplitude, and $t_{i}(n,n+1)$ is the inter-site hopping amplitude, which depends on whether the link between vertices $n$ and $n+1$ is $S$ or $L$. We will take $t_o = 2.15$ and
\begin{equation} \label{eq:t_i_choice}
    t_i(n,n+1) = \begin{cases} 3.04 & \text{ if link between sites $n$ and $n+1$ is $S$} \\ 2.73 & \text{ if link between sites $n$ and $n+1$ is $L$.} \end{cases}
\end{equation}
With these choices, the inter-site hopping term is, on average, $2.85$. This is comparable to the original SSH model \cite{su1979solitons}, where the onsite and inter-site hopping strengths are $2.15$ and $2.85$ respectively.

The quasicrystal SSH model retains the chiral symmetry of the original periodic SSH model, i.e. 
\begin{equation}
    \{ \mathcal{S} , H \} = \mathcal{S} H + H \mathcal{S} = 0
\end{equation}
where $\mathcal{S} = I \otimes \sigma_z$. We can therefore consider the model as belonging to class BDI of the Altland-Zirnbauer classification of topological insulators \cite{altland1997nonstandard,kitaev-2009}. 

It is straightforward to compute the Bloch eigenvalue bands of the periodic SSH model \cite{su1979solitons,Asboth} as
\begin{equation}
    E_\pm(k) = \sqrt{ t_o^2 + t_i^2 + 2 t_o t_i \cos(k) } \quad k \in [-\pi,\pi]
\end{equation}
so that the bulk spectrum is exactly $[-|t_o - t_i|,-|t_o + t_i|] \cup [|t_o - t_i|,|t_0+t_i|]$. The bulk winding number is 1 
whenever $|t_o| < |t_i|$. When $t_0 = 2.15$ and $t_i = 2.85$ the bulk spectrum of the periodic SSH model is therefore 
\begin{equation}
    [-5,-0.7] \cup [0.7,5],
\end{equation}
and the bulk winding number is 1.

The infinite Fibonacci SSH Hamiltonian with $t_i$ chosen according to \eqref{eq:t_i_choice} is a perturbation whose size in the operator norm is bounded by $0.19$  of the standard SSH Hamiltonian which has gap $0.7$ and topological index $1$. By standard arguments the Fibonacci SSH Hamiltonian has a gap of at least $0.51$ and must also have topological index $1$. When the model is truncated and Dirichlet boundary conditions are imposed at both ends, we expect edge states, eigenvectors of the Hamiltonian supported near to the physical edge of the model, to occur.






We can compute an approximation of the spectrum and integrated DOS for the Fibonacci SSH model by imposing periodic boundary conditions on a finite chain of vertices. These computations are shown in Figures \ref{fig:H_per_spec} and \ref{fig:H_per_IDOS}. The computations confirm that in passing from the periodic SSH model to the Fibonacci SSH model the large spectral gap at zero persists, and suggest that the spectrum of the Fibonacci SSH model has new smaller gaps appearing within the bands of the periodic model.


\begin{figure}
\begin{centering}
\includegraphics[scale=.5]{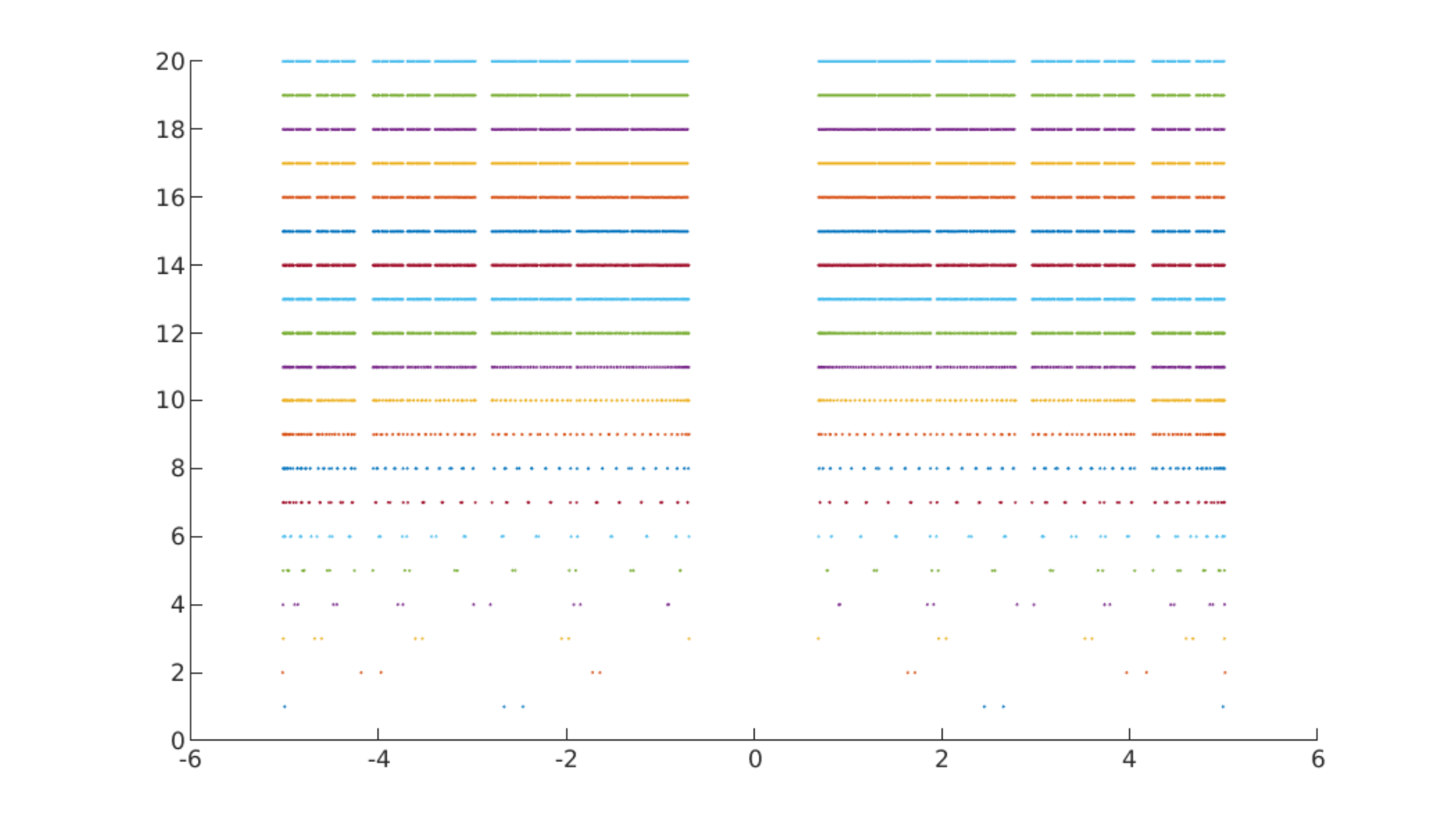}
\end{centering}
\caption{Spectra of the quasicrystal SSH model introduced in Section \ref{sec:Fib_SSH_2} computed by imposing periodic boundary conditions on a finite chain, for increasing stages of the Fibonacci quasicrystal construction (equivalently, increasing system sizes). Quasicrystal stage number \eqref{eq:all_stages} is shown on the $y$ axis, eigenvalues along the $x$ axis. The computations appear to converge to a limit spectrum with a large gap at $0$ and several smaller gaps within the bands of the periodic SSH model.}
\label{fig:H_per_spec}
\end{figure}

\begin{figure}
\begin{centering}
\includegraphics[scale=.5]{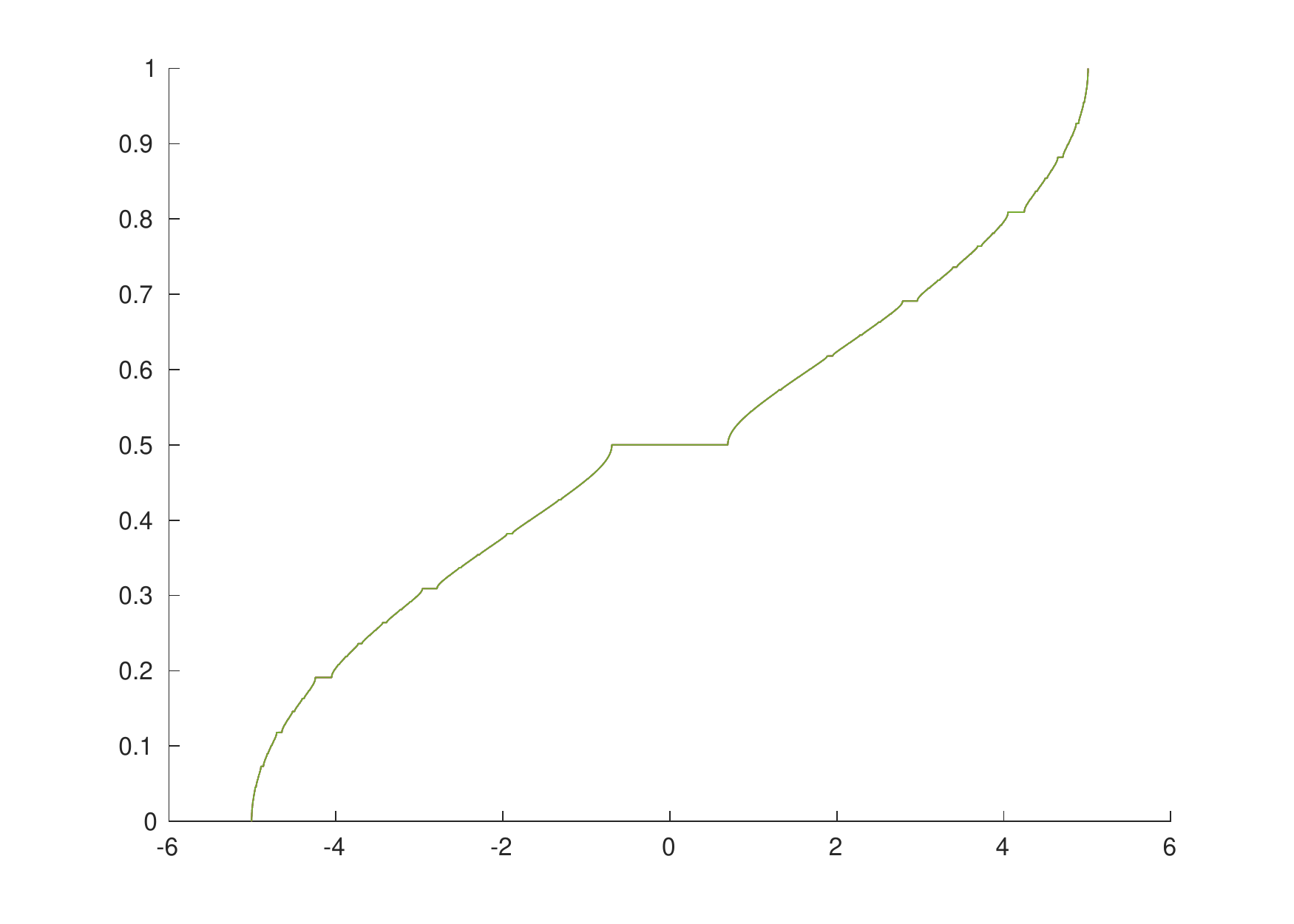}
\end{centering}
\caption{Integrated density of states (IDOS) for stages 16 through 20 of the quasicrystal SSH model introduced in Section \ref{sec:Fib_SSH_2}, computed by imposing periodic boundary conditions on the finite chain. The IDOS for successive stages are indistinguishable, demonstrating the convergence of the computations as stage number (equivalently system size) is increased. The density shows a clear gap at $0$ and smaller gaps away from zero.}
\label{fig:H_per_IDOS}
\end{figure}

\section{Investigation of the wLDOS for the Fibonacci SSH model} \label{sec:numerical_expts}

In this section we show computations of the wLDOS for the Fibonacci SSH model introduced in the previous section. For real positive $\eta > 0$, we define a Gaussian energy window with standard deviation $\eta$ by
\begin{equation} \label{eq:eta_f}
    f_{\eta}(\xi) = e^{- (\eta^{-1} \xi)^2}.
\end{equation}
To avoid diagonalization of the Hamiltonian, we approximate $f_\eta$ by a 14th order polynomial. We define a position window by
\begin{equation}
    g_2(\xi) := g(2 \xi),
\end{equation}
where $g(\xi)$ is as in \eqref{eq:f}, so that $g_2(\xi)$ is supported on the interval $[-1,1]$. This window function is plotted in Figure \ref{fig:Windows-for-later}.
With these window functions, we compute the wLDOS in the bulk (Figure \ref{fig:GLDOS_bulk_pictures}) and at the edge (Figure \ref{fig:GLDOS_end_pictures}) of finite truncations of the model. By our theoretical results we know that when the truncation is sufficiently far from the point of interest that the results will be identical to those obtained if we were able to compute with the fully infinite model (half-infinite when we look at the edge). We can computationally test locality of the wLDOS by comparing the computed wLDOS as we increase the system size (Figures \ref{fig:GLDOS_bulk_data} and \ref{fig:GLDOS_end_data}). We find that in practice the computed wLDOS converges quickly as system size is increased, suggesting the wLDOS is more local than our theoretical guarantees.

\begin{figure}
\raisebox{4cm}{(a)}\includegraphics[scale=.75]{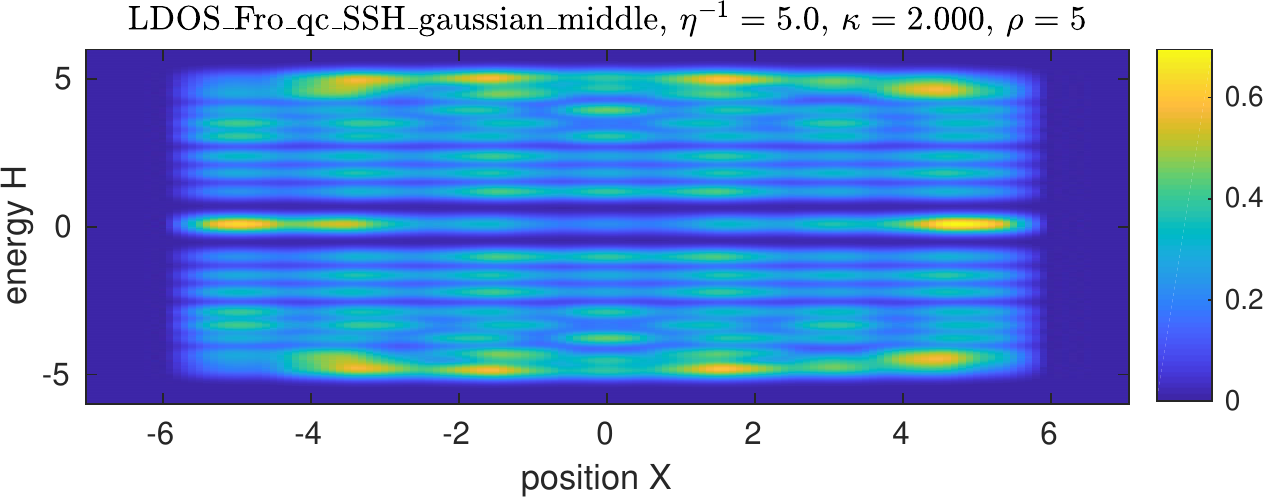}\\
\raisebox{4cm}{(b)}\includegraphics[scale=.75]{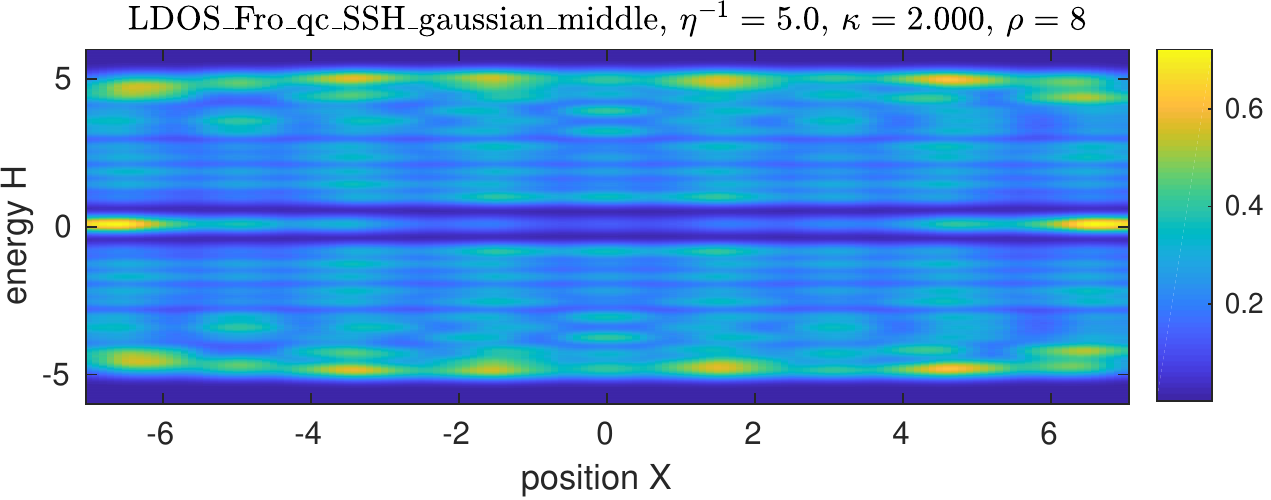}\\
\raisebox{4cm}{(c)}\includegraphics[scale=.75]{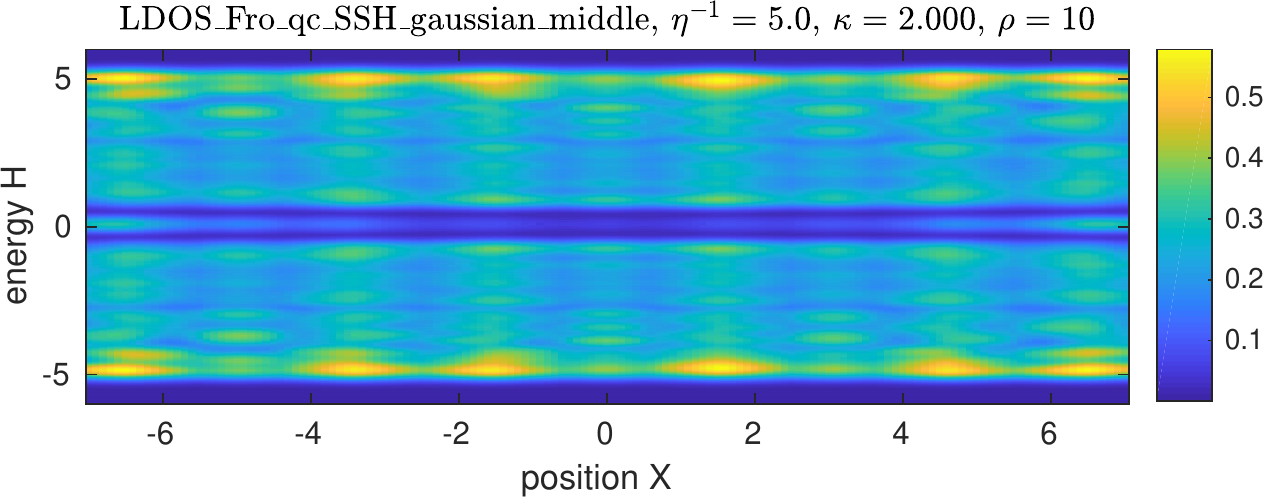}\\
\raisebox{4cm}{(d)}\includegraphics[scale=.75]{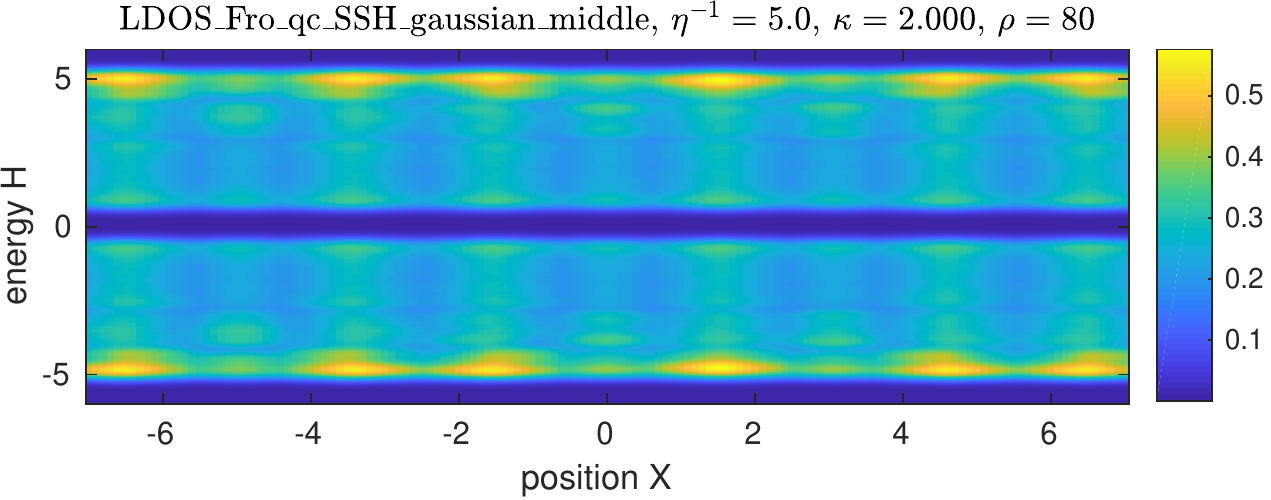}

\caption{Computed wLDOS in the center (bulk) of a finite quasicrystal SSH chain from $-7$ to $7$ for different system sizes. The inverse of the standard deviation of the energy window is fixed at $\eta^{-1}=5$. In (a), the chain extends from (roughly) $-5$ to $5$. In (b), the chain extends from $-8$ and $8$. In (c), from $-10$ to $10$. In (d), from $-80$ to $80$. Exponentially-decaying edge modes in the bulk spectral gap at $0$ are clearly visible in (a)-(c). In (d), the edges of the chain are sufficiently far away that the edge modes do not appear. The bulk gap at $0$ and smaller gaps away from zero, previously seen in Figures \ref{fig:H_per_spec} and \ref{fig:H_per_IDOS}, are clearly visible.
}
\label{fig:GLDOS_bulk_pictures}
\end{figure}

\begin{figure}
\includegraphics[scale=.75]{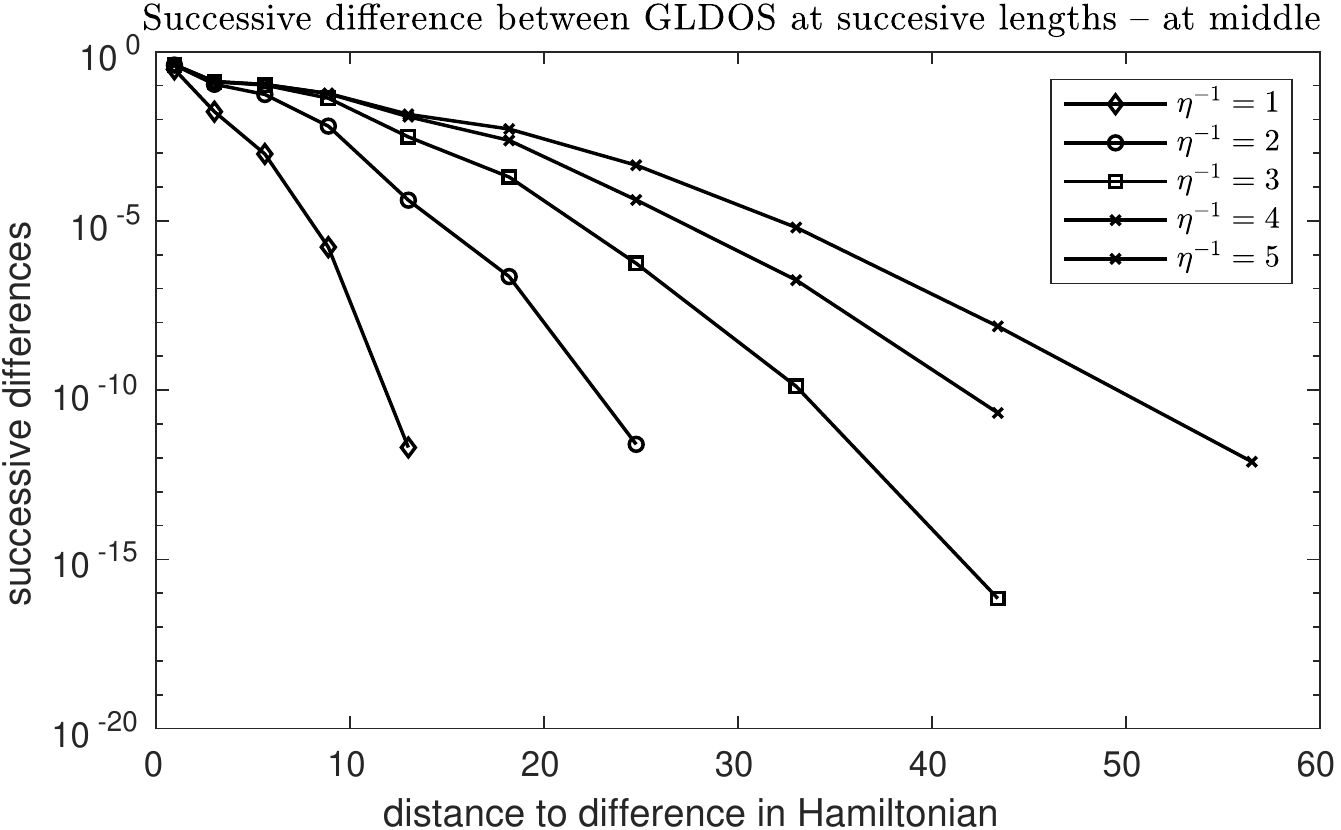}

\caption{Difference in successive terms in sequence of wLDOS computed with increasing system lengths, for energy windows of different widths. For wide energy windows (small $\eta^{-1}$), the sequence shows rather quick convergence to machine precision, as should be expected from the uncertainty principle.
}
\label{fig:GLDOS_bulk_data}
\end{figure}



\begin{figure}
\raisebox{4cm}{(a)}\includegraphics[scale=.75]{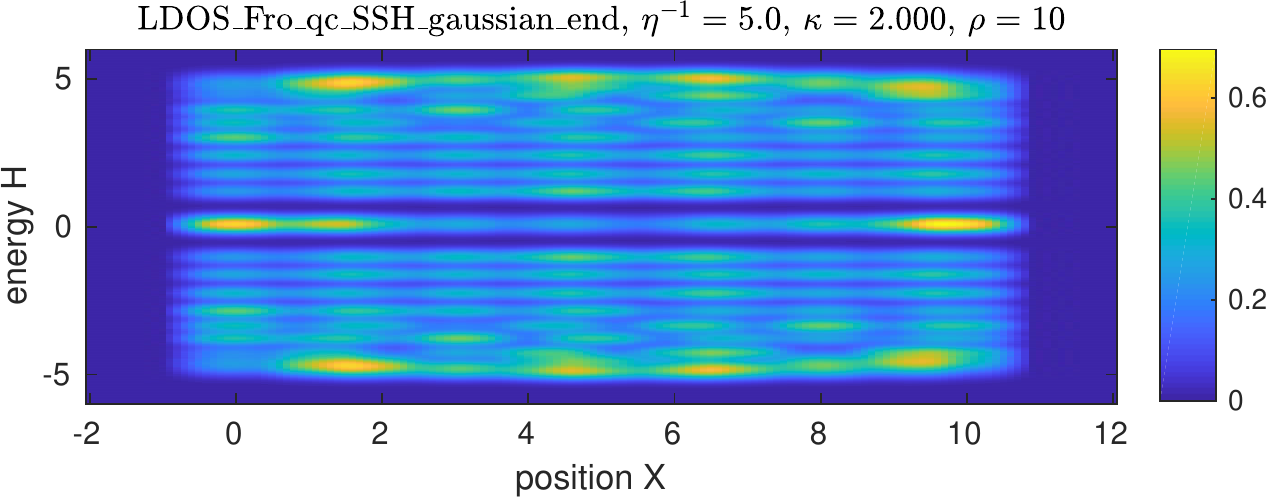}\\
\raisebox{4cm}{(b)}\includegraphics[scale=.75]{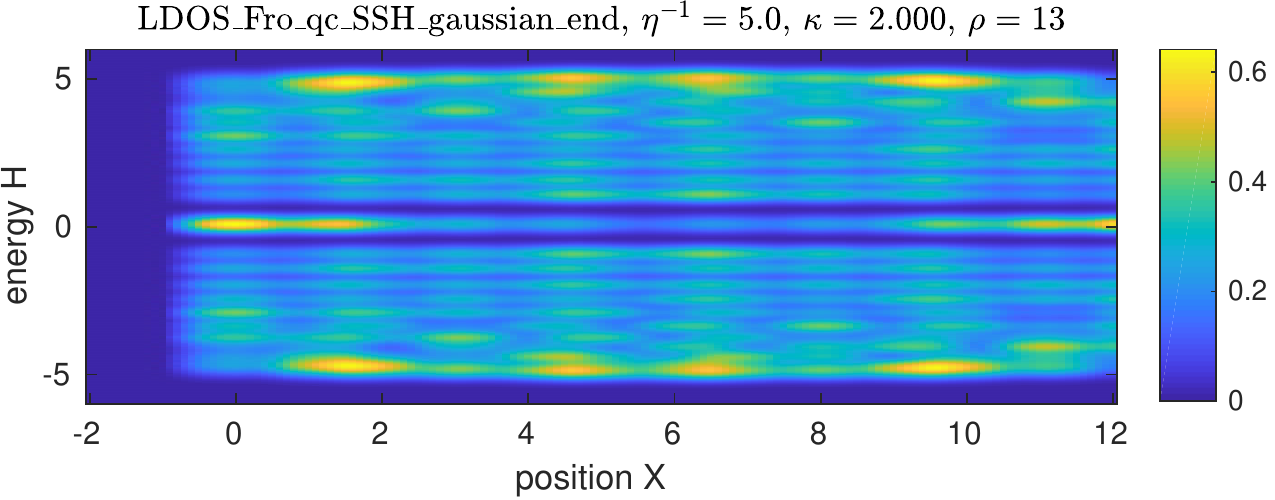}\\
\raisebox{4cm}{(c)}\includegraphics[scale=.75]{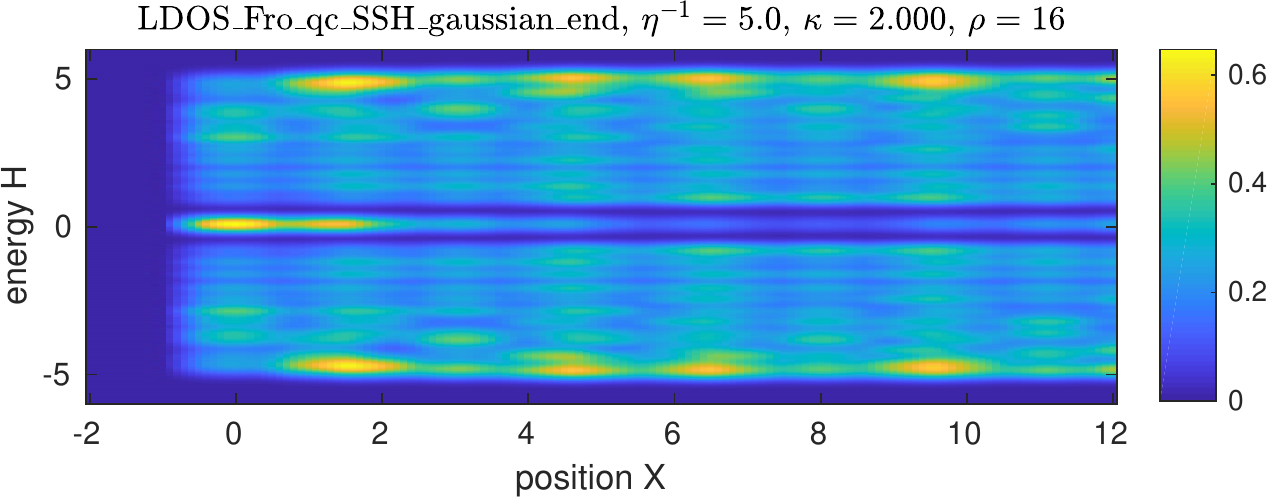}\\
\raisebox{4cm}{(d)}\includegraphics[scale=.75]{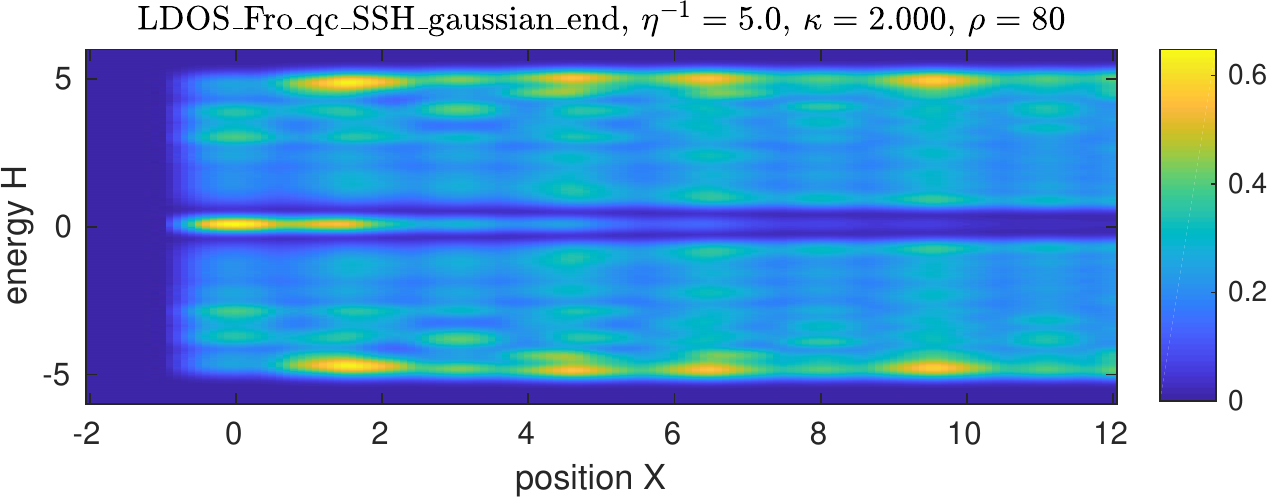}
\caption{Computed wLDOS at the edge of a finite quasicrystal SSH chain from $-1$ to $12$ for various system sizes. The inverse of the standard deviation of the energy window is fixed at $\eta^{-1} = 5$. In (a), the chain extends from (roughly) $0$ to $10$. In (b) the chain extends from $0$ to $13$. In (c), from $0$ to $16$. In (d), from $0$ to $80$. In all figures, the edge mode with energy in the bulk spectral gap can be clearly seen.
}
\label{fig:GLDOS_end_pictures}
\end{figure}

\begin{figure}
\includegraphics[scale=.75]{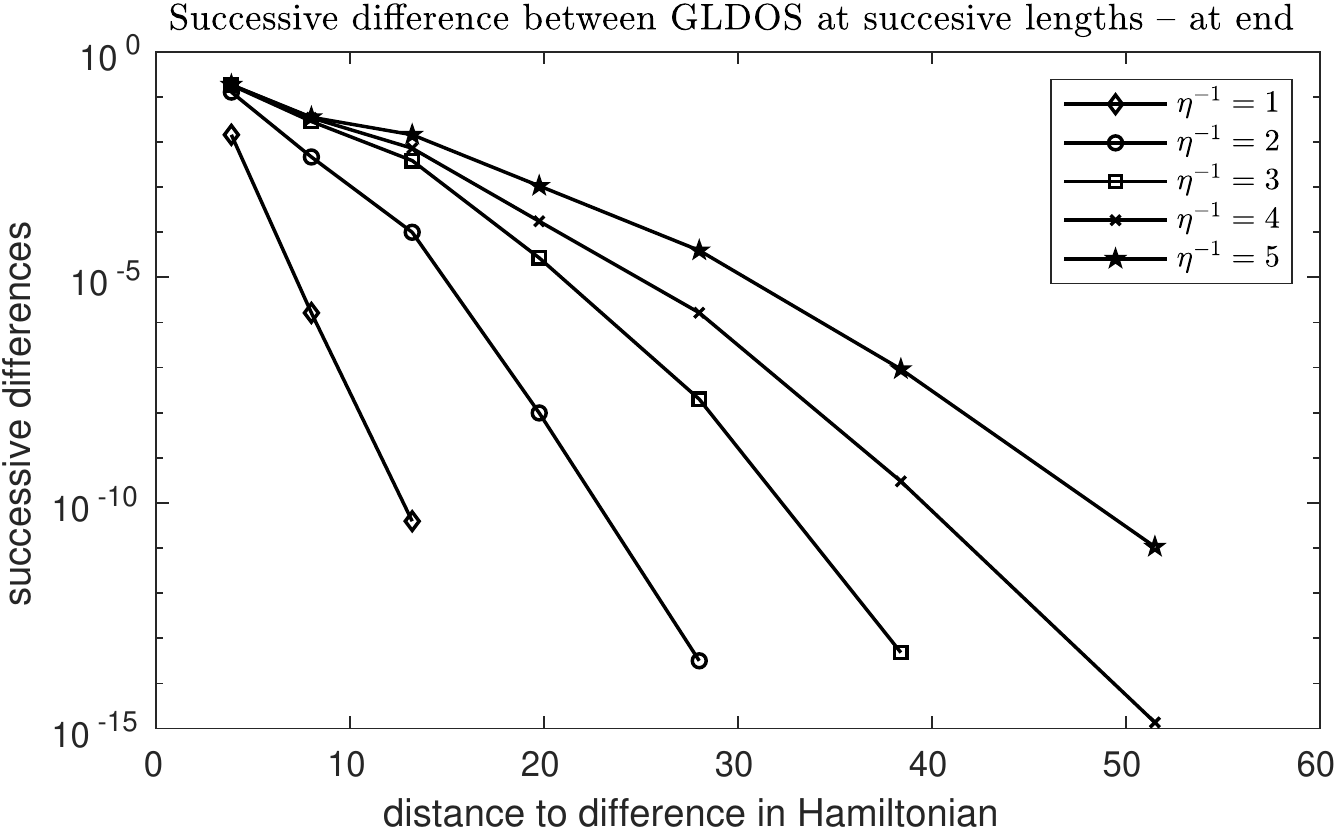}
\caption{
Difference in successive terms in sequence of wLDOS computed at edge with increasing system lengths. The sequences show rather quick convergence to machine precision, especially for wide energy windows (small $\eta^{-1}$).
}
\label{fig:GLDOS_end_data}
\end{figure}

\section{Conclusion}

In this work we have proposed a variant of the usual LDOS called the
windowed local density of states or wLDOS where the LDOS is
``windowed'' with respect to energy and position. We proved that for
finite systems the wLDOS generalizes the usual LDOS in the sense that
for narrow window functions the wLDOS reduces to the usual LDOS. We
proved that the wLDOS is local in the sense that it can be computed
accurately from a finite truncation of the Hamiltonian about the point
of interest. We used this property to show that under natural locality
assumptions on the Hamiltonian the wLDOS is well-defined and
computable even for tight-binding models on infinite domains. We
finally investigated the wLDOS for the ``Fibonacci SSH'' model, a one
dimensional aperiodic model which nonetheless has a non-trivial bulk
index and associated topological edge states. Our computations show
that the wLDOS is considerably more local than our theory guarantees,
suggesting our locality estimates might be further improved.

These results demonstrate that the wLDOS can be a useful tool for computing and visualizing material properties. Since our theory allows for computing the wLDOS for both finite and infinite models, it may be useful for predicting when finite size effects will be important in computations and/or experiments.

For future directions, we mention two potential applications of our results. First, to compute the DOS of a large finite system, instead of diagonalizing the Hamiltonian, one could average the wLDOS over the system. This computation would require many computations of the wLDOS, but importantly, using locality, these computations could be performed in parallel. This scheme could be applied, for example, to compute the DOS of a large finite quasicrystal as studied by one of the authors in \cite{2019Loring_2}. Similar schemes have been employed in the study of incommensurate layered materials, see \cite{2017MassattLuskinOrtner,doi:10.1137/17M1141035,Carr2017}. 


The second potential application involves an apparent link between
the wLDOS and the spectral localizer \cite{2011HastingsLoring,2015Loring,2019Loring}.
If one uses a broad window $g$ in position and a modest window $f$
in energy then $g(X)$ and $f(H)$ become almost commuting matrices.
Given the known connection \cite{2015Loring} between almost commuting
matrices, $K$-theory, and the spectral localizer, it should be interesting
to look at the wLDOS on topological systems. Moreover, it should be
useful to work with versions of the wLDOS using other norms,
c.f. \eqref{eq:trace_frob}.




\appendix

\section{Proofs of Lemma \ref{lem:locality_opnorm} and Proposition \ref{prop:k_alpha}} \label{app:locality_lemmas}

In this section we give the proofs of Lemma \ref{lem:locality_opnorm} and Proposition \ref{prop:k_alpha}.

\subsection{Proof of Lemma \ref{lem:locality_opnorm}}

Before we can give the proof of Lemma \ref{lem:locality_opnorm} we require two preliminary Lemmas. Note that we take $x = E = 0$ for simplicity.
\begin{lemma} \label{lem:first}
Suppose that $A$ and $B$ are bounded linear operators, with $A$ Hermitian. Then
\begin{equation}
    \| [ e^{i t A} , B ] \| \leq |t| \| [ A , B ] \|.
\end{equation}
\end{lemma}
\begin{proof}
First note that
\begin{equation}
    \fdf{t} \Bigl( e^{i t A} B e^{- i t A}\Bigr) = i e^{i t A} [ A, B ] e^{- i t A}. 
\end{equation}
Integrating from $0$ to $t$, we get 
\begin{equation} \label{eq:this}
   e^{i t A} B e^{- i t A} - B = i \inty{0}{t}{ e^{i s A} [ A, B ] e^{- i s A} }{s}.
\end{equation}
We now move to proving the estimate starting with
\begin{equation}
    \| [ e^{i t A} , B ] \| = \| e^{i t A} B - B e^{i t A} \| = \| e^{i t A} B e^{- i t A} - B \|.
\end{equation}
Using \eqref{eq:this} we now have
\begin{equation}
    \| [ e^{i t A} , B ] \| = \left\| \inty{0}{t}{ e^{i s A} [ A, B ] e^{- i s A} }{s} \right\| \leq |t| \| [A,B] \|. \qedhere
\end{equation}
\end{proof}
\begin{lemma} \label{lem:lucien}
Let $H$ and $X$ be self-adjoint operators with $H$ bounded. Let $g$ and $k$ be functions on $\mathbb{R}$ with $0 \leq g \leq 1$, $0 \leq k \leq 1$ and $k g = g$. Then
\begin{equation}
    \| g^{\frac{1}{2}}(X) e^{i t H} g^{\frac{1}{2}}(X) - g^{\frac{1}{2}}(X) e^{i t k(X) H k(X)} g^{\frac{1}{2}}(X) \| \leq |t| ( 1 + |t| \| H \| ) \| [k(X),H] \|.
\end{equation}
\end{lemma}
\begin{proof}
First note that
\begin{equation} \label{eq:main_estimate}
\begin{split}
    &\left\| g^{\frac{1}{2}}(X) e^{i t H} g^{\frac{1}{2}}(X) - g^{\frac{1}{2}}(X) e^{i t k(X) H k(X)} g^{\frac{1}{2}}(X) \right\| \\
    & \quad \leq \left\| g^{\frac{1}{2}}(X) - e^{- i t H} e^{i t k(X) H k(X)} g^{\frac{1}{2}}(X) \right\| \\
    & \quad \leq \inty{0}{t}{ \left\| \fdf{s} \left( g^{\frac{1}{2}}(X) - e^{- i s H} e^{i s k(X) H k(X)} g^{\frac{1}{2}}(X) \right) \right\| }{s} \\
    & \quad = \inty{0}{t}{ \left\| \fdf{s} \left( e^{- i s H} e^{i s k(X) H k(X)} g^{\frac{1}{2}}(X) \right) \right\| }{s}.
\end{split}
\end{equation}
Clearly,
\begin{equation} \label{eq:derivative}
    \fdf{t} \left( e^{- i t H} e^{i t k(X) H k(X)} g^{\frac{1}{2}}(X) \right) = i e^{- i t H} \left( k(X) H k(X) - H \right) e^{i t k(X) H k(X)} g^{\frac{1}{2}}(X).
\end{equation}
Re-arranging the right-hand side of \eqref{eq:derivative} and using $k g = k$ gives 
\begin{equation} \label{eq:derivative_2}
\begin{split}
    &\fdf{t} \left( e^{- i t H} e^{i t k(X) H k(X)} g^{\frac{1}{2}}(X) \right)    \\
    &\quad = i e^{- i t H} \left( [k(X),H] e^{i t k(X) H k(X)} \right) g^{\frac{1}{2}}(X)  \\
    &\quad \quad + i e^{- i t H} \left( k(X) H [ k(X) , e^{i t k(X) H k(X)} ] + H [ k(X) , e^{i t k(X) H k(X)} ] \right) g^{\frac{1}{2}}(X).
\end{split}
\end{equation}
The first term on the right-hand side of \eqref{eq:derivative_2} can be bounded by $\| [k(X),H] \|$. Using Lemma \ref{lem:first} we have
\begin{equation}
    \| [ k(X), e^{i t k(X) H k(X)} ] \| \leq |t| \| [ k(X) , k(X) H k(X) ] \| = |t| \| [ k(X) , H ] \|
\end{equation}
and hence the second two terms on the right-hand side of \eqref{eq:derivative_2} are bounded by $2 |t| \| H \| \| [ k(X), H ] \|$. Substituting these estimates into \eqref{eq:main_estimate} we have 
\begin{equation} \label{eq:main_estimate_2}
\begin{split}
    &\left\| g^{\frac{1}{2}}(X) e^{i t H} g^{\frac{1}{2}}(X) - g^{\frac{1}{2}}(X) e^{i t k(X) H k(X)} g^{\frac{1}{2}}(X) \right\| \\
    & \quad \leq \inty{0}{t}{ \| [k(X),H] \| + 2 |s| \| H \| \| [k(X),H] \| }{s}    \\
    & \quad \leq |t| ( 1 + |t| \| H \| ) \| [k(X),H] \|
\end{split}
\end{equation}
as required.
\end{proof}

We can now give the proof of Lemma \ref{lem:locality_opnorm}.

\begin{proof} [proof of Lemma \ref{lem:locality_opnorm}]
Let $\ell$ be the inverse Fourier transform of $f'$, so
\begin{equation}
    f'(\xi)=\int_{-\infty}^{\infty}\ell(t)e^{it\xi}\,dt.
\end{equation}
For bounded Hermitian operators $K$, we have by functional calculus that
\begin{equation}
    f(K)=f(0)I+\int_{-\infty}^{\infty}\frac{\ell(t)}{it}(e^{itK}-I)\,dt.
\end{equation}
It then follows that
\begin{equation}
    g^{\frac{1}{2}}(X)f(K)g^{\frac{1}{2}}(X)=f(0)g(X)+\int_{-\infty}^{\infty}\frac{\ell(t)}{it}\left(g^{\frac{1}{2}}(X)e^{itK}g^{\frac{1}{2}}(X)-g(X)\right)\,dt.
\end{equation}
Comparing this identity with $K=H$ and $K=k(X)Hk(X)$ we find
\begin{align*}
 & g^{\frac{1}{2}}(X)f(H)g^{\frac{1}{2}}(X)-g^{\frac{1}{2}}(X)f(k(X) H k(X))g^{\frac{1}{2}}(X)\\
 & \qquad=\int_{-\infty}^{\infty}\frac{\ell(t)}{it}\left(g^{\frac{1}{2}}(X)e^{itH}g^{\frac{1}{2}}(X)-g^{\frac{1}{2}}(X)e^{itk(X)Hk(X)}g^{\frac{1}{2}}(X)\right)\,dt.
\end{align*}
Therefore 
\begin{align*}
 & \left\Vert g^{\frac{1}{2}}(X)f(H)g^{\frac{1}{2}}(X)-g^{\frac{1}{2}}(X)f(k(X) H k(X) ) g^{\frac{1}{2}}(X)\right\Vert \\
 & \qquad\leq\int_{-\infty}^{\infty}\frac{\left|\ell(t)\right|}{\left|t\right|}\left( |t| ( 1 + |t| \| H \| ) \| [k(X),H] \| \right)\,dt\\
 & \qquad= \left\Vert \left[k(X),H\right]\right\Vert \int_{-\infty}^{\infty} ( 1 + | t | \| H \| ) \left| \ell(t)\right|\,dt.
\end{align*}
Since $\ell(t) = i t \widehat{f}(t)$ the statement follows.
\end{proof}

\subsection{Proof of Proposition \ref{prop:k_alpha}}

We now prove Proposition \ref{prop:k_alpha}. We again start with a preliminary Lemma.

\begin{lemma} \label{lem:second}
Suppose that $A$ and $B$ are bounded matrices, with $A$ Hermitian, and suppose $k \in C^1(\field{R})$ is such that
\begin{equation} \label{eq:f_deriv}
    k'(x) = \inty{-\infty}{\infty}{ \ell(t) e^{i t x} }{t}
\end{equation}
for some $\ell(t) \in L^1(\field{R})$. Then
\begin{equation}
    \| [ k(A), B ] \| \leq \| \ell \|_{L^1} \| [ A, B ] \|.
\end{equation}
\end{lemma}
\begin{proof}
Using \eqref{eq:f_deriv}, we have
\begin{equation}
\begin{split}
    k(x) &= k(0) + \inty{0}{x}{ k'(y) }{y}   \\
    &= k(0) + \inty{0}{x}{ \inty{-\infty}{\infty}{ \ell(t) e^{i t y} }{t} }{y}   \\
    &= k(0) + \inty{-\infty}{\infty}{ \ell(t) \inty{0}{x}{ e^{i t y} }{y} }{t}    \\
    &= k(0) + \inty{-\infty}{\infty}{ \ell(t) \left( \frac{ e^{i t x} - 1 }{ i t } \right) }{t}.
\end{split}
\end{equation}
Hence
\begin{equation} \label{eq:f_2}
    k(A) = k(0) I + \inty{-\infty}{\infty}{ \frac{ \ell(t) }{ i t } \left( e^{i t A} - 1 \right) }{t},
\end{equation}
where the integral is well-defined because
\begin{equation}
    \| e^{i t A} - 1 \| \leq | t | \| A \|
\end{equation}
(to see this differentiate the operator on the left-hand side). From \eqref{eq:f_2}, we have that
\begin{equation}
    [ k(A) , B ] = \inty{-\infty}{\infty}{ \frac{ \ell(t) }{ i t } [ e^{i t A} , B ] }{t}.
\end{equation}
The result now follows by Lemma \ref{lem:first}.
\end{proof}
We are now in a position to give an explicit construction of $k_\alpha(\xi)$ equalling $1$ for all $\xi \in [-L,L]$ and satisfying \eqref{eq:kXH}.
\begin{proof}[Proof of Proposition \ref{prop:k_alpha}]
Fix $L > 0$, and define $M_L$ to be the smallest integer such that $2 M_L > L$. Define $k(\xi)$ as in \eqref{eq:f}, i.e.
\begin{equation} \label{eq:f_3}
    k(\xi) = \begin{cases} 0 & \xi \leq -2 \\ \frac{1}{2} (\xi + 2)^2 & -2 < \xi \leq -1 \\ 1 - \frac{1}{2} \xi^2 & -1 < \xi \leq 1 \\ \frac{1}{2} (\xi - 2)^2 & 1 < \xi \leq 2 \\ 0 & 2 \leq \xi. \end{cases}
\end{equation}
It is clear that
\begin{equation}
    k(\xi) + k(\xi-2) = \begin{cases} 0 & \xi \leq -2 \\ \frac{1}{2}(\xi+2)^2 & -2 < \xi \leq -1 \\ 1 - \frac{1}{2}\xi^2 & -1 < \xi \leq 0 \\ 1 & 0 < \xi \leq 2 \\ 1 - \frac{1}{2} (\xi - 2)^2 & 2 \leq \xi < 3 \\ \frac{1}{2}(\xi - 4)^2 & 3 < \xi \leq 4 \\ 0 & 4 < \xi, \end{cases}
\end{equation}
and more generally, 
\begin{equation}
    k_1(\xi) := \sum_{l = -M}^M k(\xi - 2 l) = \begin{cases} 0 & \xi \leq - 2 M - 2 \\ \frac{1}{2}(\xi + 2 M + 2)^2 & - 2 M - 2 \leq \xi \leq - 2 M - 1 \\ 1 - \frac{1}{2}(\xi + 2M)^2 & - 2 M - 1 \leq \xi \leq - 2 M \\ 1 & -2 M \leq \xi \leq 2 M \\ 1 - \frac{1}{2} (\xi - 2 M)^2 & 2 M \leq \xi \leq 2 M + 1 \\ \frac{1}{2}(\xi - 2 M - 2)^2 & 2 M + 1 \leq \xi \leq 2 M + 2 \\ 0 & \xi \geq 2 M + 2. \end{cases}
\end{equation}
Using the fact that $2 M > L$ by assumption, we have that $k_1(\xi)$ acts by $1$ over the whole interval $[-L,L]$. For positive $\alpha$, we now define
\begin{equation}
    k_\alpha(\xi) := k_1(\alpha \xi).
\end{equation}
It is easy to see that $k_\alpha(\xi)$ acts as $1$ over the interval $\left[-\frac{2 M}{\alpha},\frac{2 M}{\alpha}\right]$ and hence acts as $1$ over the interval $\left[-\frac{L}{\alpha},\frac{L}{\alpha}\right]$, which contains $[-L,L]$ for $0 < \alpha < 1$. The support of $k_\alpha(\xi)$ is clearly confined to $\left[ - \frac{2M+2}{\alpha} , \frac{2 M + 2}{\alpha} \right]$. Using the definition of $M$ as the smallest integer such that $2 M > L$ we see that $L \leq 2 M \leq L + 2 \leq 2 M + 2 \leq L + 4$ and hence the support of $k_\alpha(\xi)$ is confined to $\left[ - \frac{ 2 M + 4 }{ \alpha } , \frac{ 2 M + 4 }{ \alpha } \right]$.

We will now prove that $k_\alpha(\xi)$ satisfies the bound \eqref{eq:kXH} using Lemma \ref{lem:second}. Our strategy is to build up to a bound on the Fourier transform of $k_\alpha'(\xi)$ from a bound on the Fourier transform of $k'(\xi)$.

We start by noting that if $k(\xi)$ is defined by \eqref{eq:f_3}, then
\begin{equation}
    k'(\xi) = \begin{cases} 0 & \xi \leq -2 \\ \xi + 2 & -2 < \xi \leq -1 \\ - \xi & -1 < \xi \leq 1 \\ \xi - 2 & 1 < \xi \leq 2 \\ 0 & 2 \leq \xi. \end{cases}
\end{equation}
Since $k'(\xi)$ is odd, its Fourier transform $\ell(t)$ equals
\begin{equation}
    \ell(t) := - \frac{i}{\pi} \inty{0}{\infty}{ k'(\xi) \sin(t \xi) }{\xi} = - \frac{i}{\pi} \left[ \inty{0}{1}{ (- \xi) \sin(t \xi) }{\xi} + \inty{1}{2}{ (\xi - 2) \sin(t \xi) }{\xi} \right].
\end{equation}
Integrating by parts in these integrals we have
\begin{equation}
    \inty{0}{1}{ (- \xi) \sin(t \xi) }{\xi} = \frac{ \cos(t) }{ t } - \frac{ \sin(t) }{ t^2 }
\end{equation}
\begin{equation}
    \inty{1}{2}{ (\xi - 2) \sin(t \xi) }{\xi} = - \frac{ \cos(t) }{ t } + \frac{ \sin(2 t) }{ t^2 } - \frac{ \sin(t) }{ t^2 }
\end{equation}
and hence
\begin{equation} \label{eq:g}
    \ell(t) = - \frac{i}{\pi} \frac{ \sin(2 t) - 2 \sin(t) }{ t^2 },
\end{equation}
which is clearly in $L^1(\field{R})$. In fact, numerical computation shows that $\| \ell \|_1 \approx 1.27$ (3sf). Now let $\ell_1(t)$ denote the Fourier transform of $k_1'(\xi)$. Using linearity and a change of variables we have
\begin{equation}
    \ell_1(t) = \sum_{l = - M}^M e^{- 2 i l t} \ell(t).
\end{equation}
Using the triangle inequality we have
\begin{equation}
    \| \ell_1 \|_{L^1} \leq (2 M + 1) \| \ell \|_{L^1}.
\end{equation}
Finally, let $\ell_\alpha(t)$ denote the Fourier transform of $k_\alpha'(\xi)$. Since $k_\alpha'(\xi) = \alpha k_1'(\alpha \xi)$, we see that
\begin{equation}
    \ell_\alpha(t) = \frac{1}{2 \pi} \inty{-\infty}{\infty}{ \alpha k_1'(\alpha \xi) e^{- i t \xi} }{\xi} = \ell_1 \left(\frac{t}{\alpha}\right),
\end{equation}
from which it follows immediately that
\begin{equation}
    \| \ell_\alpha(t) \|_{L^1} = \alpha \| \ell_1(t) \|_{L^1} \leq (2 M + 1) \alpha \| \ell(t) \|_{L^1}.
\end{equation}
Applying Lemma \ref{lem:second} now proves Proposition \ref{prop:k_alpha}. 
\end{proof}

\printbibliography

\end{document}
